%% file: AD-JHEP.tex
\documentclass[a4paper,11pt]{article}
\pdfoutput=1 

\usepackage{jheppub} 

\usepackage[T1]{fontenc} 

\usepackage{graphicx, bm, amsfonts, amssymb, amsmath, amsthm, color}

\usepackage{mathtools}
\newcommand{\Bg}{\prescript{\gamma}{}B}
\newcommand{\Gg}{\prescript{\gamma}{}G}
\newcommand{\Ggf}{\prescript{\gamma}{}{\tilde G}_f}

\newcommand{\WF}{\mathbb{WF}}

\newcommand{\wec}[1]{{\bf #1}}
\newcommand{\akt}{{\boldsymbol{\act}}}

\include{Macros}

\title{\boldmath Holonomy spin foam models: Asymptotic geometry of the partition function.}

\author[a]{Frank Hellmann,}
\author[a,b,c]{Wojciech Kami{\'n}ski.}

\affiliation[a]{MPI for Gravitational Physics, Golm, Germany}
\affiliation[b]{University of Warsaw, Warsaw, Poland}
\affiliation[c]{Perimeter Institute for Theoretical Physics, Waterloo ON, Canada}

\emailAdd{Frank.Hellmann@aei.mpg.de}
\emailAdd{Wojciech.Kaminski@fuw.edu.pl}

\abstract{We study the asymptotic geometry of the spin foam partition function for a large class of models, including the models of Barrett and Crane, Engle, Pereira, Rovelli and Livine, and, Freidel and Krasnov.

The asymptotics is taken with respect to the boundary spins only, no assumption of large spins is made in the interior. We give a sufficient criterion for the existence of the partition function. We find that geometric boundary data is suppressed unless its interior continuation satisfies certain accidental curvature constraints. This means in particular that most Regge manifolds are suppressed in the asymptotic regime. We discuss this explicitly for the case of the configurations arising in the 3-3 Pachner move.

We identify the origin of these accidental curvature constraints as an incorrect twisting of the face amplitude upon introduction of the Immirzi parameter and propose a way to resolve this problem, albeit at the price of losing the connection to the SU(2) boundary Hilbert space.

The key methodological innovation that enables these results is the introduction of the notion of wave front sets, and the adaptation of tools for their study from micro local analysis to the case of spin foam partition functions.}

\begin{document}
\maketitle
\flushbottom

\section{Introduction}
Spin foam models\footnote{See \cite{Baez:1997zt,Baez:1999sr} for the origin of the name and \cite{Perez:2012wv,Bianchi2013,Barrett:2011zjj} for extensive reviews and an overview of recent developments.} are a class of discrete models for quantum space-times. The oldest and prototypical example is the Ponzano-Regge model for 3d Euclidean gravity without cosmological constant \cite{ponzanoregge,Barrett:2006ru,Freidel:2005bb,Freidel:2004vi,Freidel:2004nb,Barrett:2008wh} and the Turaev-Viro model for 3d Euclidean gravity with negative cosmological constant \cite{Turaev1992}.

The 3d spin foam models, while formulated as discrete state sums, are actually independent of the discretisation. The continuum limit is straightforward, and they define a continuum theory. In 4d a number of models have been proposed as discretisations of gravity over the years. These depend strongly on the discretisation and the continuum limit remains a key open question. The first such model was the one of Barrett and Crane (BC) \cite{Barrett:1997gw,Barrett:1999qw}. This was followed by the models of Engle, Livine, Pereira and Rovelli (EPRL) \cite{Engle:2007wy,Pereira:2007nh,Engle:2007qf,Engle:2007uq}, and that of Freidel and Krasnov (FK) \cite{Freidel:2007py,Livine:2007ya,Livine:2007vk}. A large number of variations for these models exist, notably the one of Kaminski, Kieselowski and Lewandowski (KKL) \cite{Kaminski2010,Kaminski2010a,Kaminski2012}. More recently new spin foam models were proposed in terms of holomorphic functions on spinors \cite{Dupuis:2011wy,Dupuis:2011dh,Dupuis:2011fz}, and the flux representation \cite{Baratin:2011hp,Baratin2011}.

There are many equivalent formulations of spin foam models that emphasise different aspects. The form most studied is a weighted state sum. The structure of these weights has been the subject of extensive research and is by now very well understood. In this formulation the weights can be interpreted in terms of geometric quantisation. The limit of large quantum numbers is given by geometric asymptotics.

The asymptotics of the BC weights were studied analytically \cite{Barrett:1998gs,Barrett:2002ur,Freidel:2002mj} as well as numerically \cite{Baez:2002rx,Christensen2010,Christensen2009,Christensen2002}. It was found that in this case the asymptotics contain a geometric sector but are dominated by solutions of a topological theory. The asymptotics of the weights of the EPRL/FK type state sum were computed in the sequence of papers \cite{Barrett:2009cj,Barrett:2009gg,Barrett:2009mw,Barrett:2009as,Barrett:2010ex,Conrady:2008mk,Conrady:2009px}. There it was found that, for appropriate boundary data, the weights are peaked on geometric data in the limit of large quantum numbers, and that the amplitude is given by the discrete form of the Einstein Hilbert action, the Regge action. In \cite{Barrett:2009as} it was shown that the topological theory that arose in the BC model is still present, and the boundaries that correspond to it are not suppressed. Boundary data that is neither geometric nor topological is exponentially suppressed in the asymptotic regime.

Based on these results the asymptotic behaviour of several observables attached to a single weight was obtained in the BC model \cite{Rovelli2006,Bianchi2006,Alesci2007,Alesci2008}, as well as in the new models
\cite{Alesci2009a, Bianchi2009}. It is, however, unlikely that these type of calculations truly give information on any kind of general relativity limit, see for example the discussion in \cite{Dittrich:2008pw,Hellmann:2011jn,Dittrich:2012qb,Dittrich:2012jq}. The behaviour of the entire state sum has remained largely unstudied, the exception being the results of \cite{Bonzom:2009hw}, which indicate that some of the desirable geometric properties are lost once the entire state sum is considered, in particular the geometries that occur semi-classically are flat.

In this paper we will use the formulation of spin foams in terms of holonomies established in \cite{Bahr:2010bs,Bahr:2011aa,Bahr2013,Dittrich2013}. This formulation, called holonomy spin foam models casts them as generalisations of lattice gauge theory. It is in the tradition of the formulations that arose in the context of the group field theory formulation \cite{DePietri2000,Reisenberger:2000zc} and in particular those of Oeckl and Pfeiffer \cite{Oeckl:2000hs,Pfeiffer:2001ig,Pfeiffer2002,Oeckl2003,Zapata:2004xq,Oeckl:2005rh}. As we will see in this paper this formulation is perfectly suited to study the geometricity of spin foam models at the level of the entire partition function.

The results and methods of this paper were first announced in \cite{HellKami}. Very recently Han independently found the same results using different methods \cite{Han2013a,Han2013b}.

\subsection{Spin foam models}

The prototypical example for spin foam models is the lattice quantisation of Horowitz' background field (BF) theory \cite{Horowitz:1989ng} as discussed in \cite{Baez:1997zt,Baez:1999sr}. The field content of BF theory in d space-time dimensions, with gauge group $G$ and Lie algebra $\lalg{g}$, is a $\lalg{g}$ valued d-2 form $B$ and a $\lalg{g}$-connection $\omega$ with curvature $F[\omega]$. Its action is given by

\[S_{BF}(B,\omega) = \int \tr B \wedge F(\omega).\]

The field $B$ acts as a Lagrange multiplier, the equations of motion are simply that $\omega$ is flat, and the field $B$ is covariantly constant. Discretizing this theory on a lattice naively leads to the spin foam model for $BF$ theory. For $\mathrm{G}$ equal to $\SU(2)$ and $d = 3$ we obtain the Ponzano-Regge model, the general model was given by Ooguri in \cite{Ooguri1992a}.

In order to write out this naive discretisation we introduce a number of structures. Consider a triangulated 4-manifold $M$. Now take the dual of this triangulation, that is, a vertex for every 4-simplex, an edge for every 3-simplex, a face for every 2-simplex and so on. This collection of faces, edges and vertices forms the 2-complex on which we will discretise our theory. We call the complex $\CC$, and the sets of vertices, edges and faces $\CC_v$, $\CC_e$ and $\CC_f$ respectively.

We will need a set of fiducial orientations and base points for every face of the 2-complex. We encode this by writing every face as an ordered list of the edges and vertices it contains, in the order determined by the orientation, starting and finishing at the base vertex. That is, every $f \in \CC_f$ is a list \[f = (v,e,v',e', \dots, e'',v).\] Adjacency amongst faces, edges and vertices is then written as $e,v \in f$, and analogously $v \in e$, and we write $(a,b,c) \subset f$ if $f$ contains $(a,b,c)$ as an uninterrupted list.

We can now give the naive discretisation of $BF$ theory. Our starting point is the formal path integral \begin{equation}\int[dB][d\omega] \exp(iS_{BF}) = \int[d\omega] \delta(F(\omega)).\end{equation}
We discretise the connection along the half edges $e \in \CC_e$, by taking $g_{ev} \in G$\footnote{We summarise our conventions for group elements, Lie algebra elements and coherent states in appendix~\ref{sec:conv}} to be the parallel transport from $v$ to $e$. We then define the face holonomy \[g_f = \prod_{(a,b) \subset f} g_{ab},\] where we use the convention $g_{ab} = g_{ba}^{-1}$. The path integral is then discretised as \begin{equation}\ZZ(\CC) = \int \prod_{v \in e} [d g_{ev}] \prod_f \delta(g_f).\label{eq:BF-prod_of_delta}\end{equation}

This product of delta functions forces the holonomy around every contractible loop on $\CC$ to be the identity. For finite groups it is well defined and gives, up to some normalisations, the number of group homomorphisms from the fundamental group of $M$ to $G$. For continuous groups the product of delta functions is a priori ill defined \cite{Bonzom:2010zh,Bonzom:2011br,Bonzom:2012mb,Bonzom:2010ar}. Surprisingly, Bahr showed that it can be turned into a genuine measure on the space of flat connections on $M$ \cite{Bahr:toappear}.

Note that if we replace the $\delta$ with regular plaquette weights $\omega$, equation \eqref{eq:BF-prod_of_delta} has exactly the form of the partition function of an ordinary lattice gauge theory. For example with $\omega$ chosen as heat kernels, we would obtain the standard heat kernel version of lattice Yang-Mills theory \cite{Oeckl:2005rh}. The spin foam models we want to consider arise as a generalisation of this form obtained by modifying the theory on the edges, while keeping the face weight distributional. This is motivated by the fact that the $BF$ action closely resembles the Palatini form of the action for general relativity. In fact, specialising to $d=4$, $G=\SO(4)$ or $G=\SO(3,1)$, we want to restrict $B$ to be of the form $\hodge e\wedge e$ for some vierbein $e$, as \[S_{P}(e,\omega) = S_{BF}(\hodge e\wedge e, \omega)\;.\]

The various ways to implement this restriction in the partition functions are discussed at length in the literature that defines the models we study. For the classical case we refer the reader to the original work of Capovilla et.al. \cite{Capovilla:1989ac,Capovilla:1991qb} and the recent work by Krasnov \cite{Krasnov:2009iy,Delfino2012,Krasnov2012a,Krasnov2011a,Krasnov2011b,Krasnov2010a}. As we discussed in \cite{Bahr2013,Dittrich2013} there is a wide class of models for which trying to implement this restriction at the discrete level \eqref{eq:BF-prod_of_delta} is parametrized by a single distribution $E$ on the group manifold. These models can then be most naturally written by using the twisted face holonomy \begin{equation}\tilde{g}_f = \prod_{(a,b,c) \subset f, b \in \CC_e} g_{ab}g_{bf}g_{bc}\;.\label{eq:twisted-face-hol}\end{equation} More elliptically this is just \[\tilde{g}_f = g_{ve}g_{ef}g_{ev'}g_{v'e'}\dots g_{v''e''}g_{e''f}g_{e''v}\;.\]
The models then are of the form \begin{equation}\ZZ(\CC) = \int \prod_{v \in e} [d g_{ev}] \prod_{e \in f} [d g_{ef}] E(g_{ef}) \prod_f \omega(\tilde{g}_f)\;.\end{equation}

This is the form of the partition function we will study in this paper. Note that for $E = \delta$ the additional group elements drop out and the partition function reduces to \eqref{eq:BF-prod_of_delta}. The fact that the $\omega$ at the face now operates on $\tilde{g}_f$ instead of on $g_f$ means that $g_f$ can now potentially vary away from the identity, even if $\omega = \delta$.

In this formulation it is easiest to include boundaries through the universal boundary Hilbert space of \cite{Dittrich2013}. This is achieved by specifying a boundary graph $\Gamma$ in $\CC$, and dropping the $g_{ev}$ integrations for the edges belonging to $\Gamma$, which we denote by $g^\Gamma_{ev}$:

 \begin{equation}\ZZ(\CC,\Gamma)[g^\Gamma_{ev}] = \int \prod_{\substack{v \in e\\ e \notin \Gamma_e}} [d g_{ev}] \prod_{e \in f} [d g_{ef}] E(g_{ef}) \prod_f \omega(\tilde{g}_f).\label{eq:Hol-SFM-partfunc}\end{equation}

In this paper, we will study the behaviour of this partition function in the limit of large spins in the boundary Hilbert space.

\input{AD-JHEP-Summary.tex}


\input{AD-JHEP-Wojciech.tex}


\section{Geometricity in Bivector and Holonomy Variables}\label{sec:Geometry}

In the previous section the wave front sets for the BC, FK and the EPRL model were derived in terms of equations for the interior bivector and holonomy variables of $\WW(\CC,\Gamma)$. We know that the large spin asymptotics of the vertex weights can give a geometric interpretation to bivector variables \cite{Conrady:2008mk,Barrett:2009gg,Hellmann:2010nf}. We can easily identify the vertex equations necessary for this reconstruction in the wave front sets above, however, the meaning of the $g_{ef}$ and the various holonomy equations is a priori unclear.

In this section we will give the extended geometricity results that include these variables. To do so, we will first derive bivectors and holonomies from a given geometry, noting what equations they satisfy. We will then discuss to what extend a geometry can be reconstructed given bivectors and group elements that satisfy these equations. The result will be a set $\GG$ of variables satisfying equations that contains a geometric sector. Both its variables and equations closely resemble those of $\WW$. In the section after this one will study the precise relationship between these sets.

\subsection{Holonomies and Bivectors from Geometry}\label{sec:Hol and Biv from Geo}

We give a set of 4-simplices $\sigma^v$ in $\R^4$ associated to the vertices of $\CC$, such that the pullback of the flat metric and the standard orientation on $\R^4$ to the triangulated manifold dual to $\CC$ defines an orientation and a simplexwise flat, continuous, non-degenerate geometry. These simplices have boundary tetrahedra $\tau^v_e$, with outward normals $N^v_e$. The tetrahedra intersect at the triangles $t^v_{ee'}$. At these triangles we have area outward normals $A^v_{ee'}$ in the plane of the tetrahedron $\tau^v_e$, which satisfy $A^v_{ee'}\cdot t^v_{ee'} = A^v_{ee'}\cdot N^v_e = 0$ and $|A^v_{ee'}| = |t^v_{ee'}|$.

At the middle of the edges $e$ we now introduce a tetrahedron $\tau^e$ with the same geometry as $\tau^v_e$ and $\tau^{v'}_{e}$, normal to some \[N_v^e = -N_{v'}^e\] chosen such that the orientation it inherits from the standard orientation on $\R^4$ by reduction with the normal $N^e_v$ matches that of the orientation that $\tau^{v}_e$ inherits from $N^{v}_e$. This is possible as we required the 4-simplices to define a consistent orientation on the manifold.  We also again have triangles and area normals, that are now, however, indexed by a face, and called $t^e_f$ and $A^e_f$.

\paragraph{Holonomies:} We can now define the holonomies of the discrete connection, $G_{ev} \in \SO(4)$ by requiring that

\begin{equation}
\label{eq:G_{ev-def1}} G_{ev} \tau^v_e = \tau^e_v.
\end{equation}

From this we immediately have that, for $f$ being the face to which $e$, $v$, and $e'$ are adjacent, the outward normals behave naturally:

\begin{align}\label{eq:G_ev-def2}
G_{ev} N^v_e & = N^e_v, \nn\\
G_{ev} A^v_{ee'} &= A^e_f.
\end{align}

We also define $G_{ev} = G_{ve}^{-1}$. We can now also introduce the simplicity rotations $G_{ef}$. These are the interior dihedral rotations of the 4-simplices in the frames of the tetrahedra at the edges. Let $e'$ precede $v$ precede $e$ in the fiducial orientation around the face $f = (\dots e, v, e' , \dots)$, then we define:

\begin{align}\label{eq:G_ef-def}
G_{ef} t^e_f & = t^e_f, \nn\\
- G_{ve} G_{ef} G_{ev} N^v_{e'}&= N^v_e.
\end{align}

\begin{figure*}[htpb]
 \centering
 \includegraphics[scale=.8]{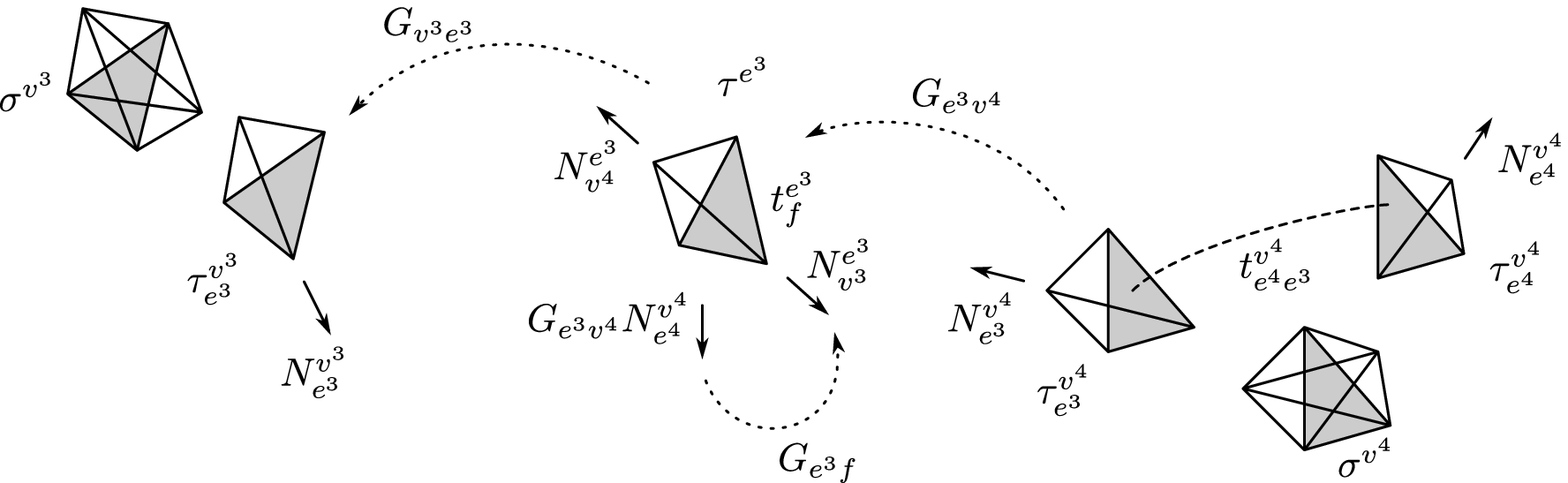}
 \caption{Geometric quantities at edge $e^3$ in the face $f$, with the sequence of vertices and edges determined by the fiducial orientation being $f = (\dots, e^4,v^4,e^3,v^3, \dots)$.}
 \label{fig-data2}
\end{figure*}

The last line can equivalently be written as $G_{ef} G_{ev} N^v_{e'}= N^e_{v'}$, for $v'$ the vertex succeeding $e$ in the fiducial orientation, that is $f = (\dots v', e, v, e' , \dots)$. It follows that we have $G_{v'e} G_{ef} G_{ev} N^v_{e'}= N^{v'}_{e}$. This arrangement of group elements is given pictorially in figure \ref{fig-data2}.

We then define two group elements around a face $f$ with $n$ edges $e \dots e^{(n)}$ and orientation $f = (v,e^{(n)},v^{(n-1)},e^{(n-1)},\dots,v',e,v)$. The first is $G_f$:

\begin{align}\label{eq:tildeG_f-def} G_{f} &= \prod_{(a,b) \subset f} G_{ab}\nn\\ &= G_{ve^{(n)}} G_{e^{(n)}v^{(n-1)}} G_{v^{(n-1)}e^{(n-1)}} G_{e^{(n-1)}v^{(n-2)}} \dots G_{v'e} G_{ev}. \end{align}

This is the holonomy around the face $f$. From the definitions it is immediate that it satisfies

\begin{equation}\label{eq:tildeG_f-stabilityTriag} G_f t^v_{ee^{n}} = t^v_{ee^{n}},\end{equation}

but

\[ G_f N^v_{e} \neq N^v_{e}.\]

The angle between $G_f N^v_{e}$ and $N^v_{e}$ is exactly the deficit angle around the face $f$,
\begin{equation}
\label{eq:DeficitAngle}\Theta_f = 2\pi - \sum_{e \subset f} \xi_{ef},
\end{equation}
where $\xi_{ef}$ are the interior dihedral angles between $e'$ and $e$ at the vertex $v$ in $f = (\dots ,e,v,e', \dots)$ encoded in the interior dihedral rotations $G_{ef}$.
The other group element we can define is the analogue of the one that appears in the partition function,

\begin{align}\label{eq:G_f}\tilde G_{f} &= \prod_{(a,e,b) \subset f, e \in \CC_e} G_{ae} G_{ef} G_{eb}\nn\\ &= G_{ve^{(n)}} G_{e^{(n)}f} G_{e^{(n)}v^{(n-1)}} G_{v^{(n-1)}e^{(n-1)}} G_{e^{(n-1)}f} G_{e^{(n-1)}v^{(n-2)}} \dots G_{v'e} G_{ef} G_{ev}. \end{align}

This has the property that \[\tilde G_f t^{v}_{e(n)e} = t^{v}_{e(n)e}\] and that \[\tilde G_f N^{v}_{e(n)} = N^{v}_{e(n)},\] thus we have that \begin{equation}\label{eq:G_f=id} \tilde G_f = \id. \end{equation}

\begin{figure}[htp]
 \centering
 \includegraphics[scale=.70]{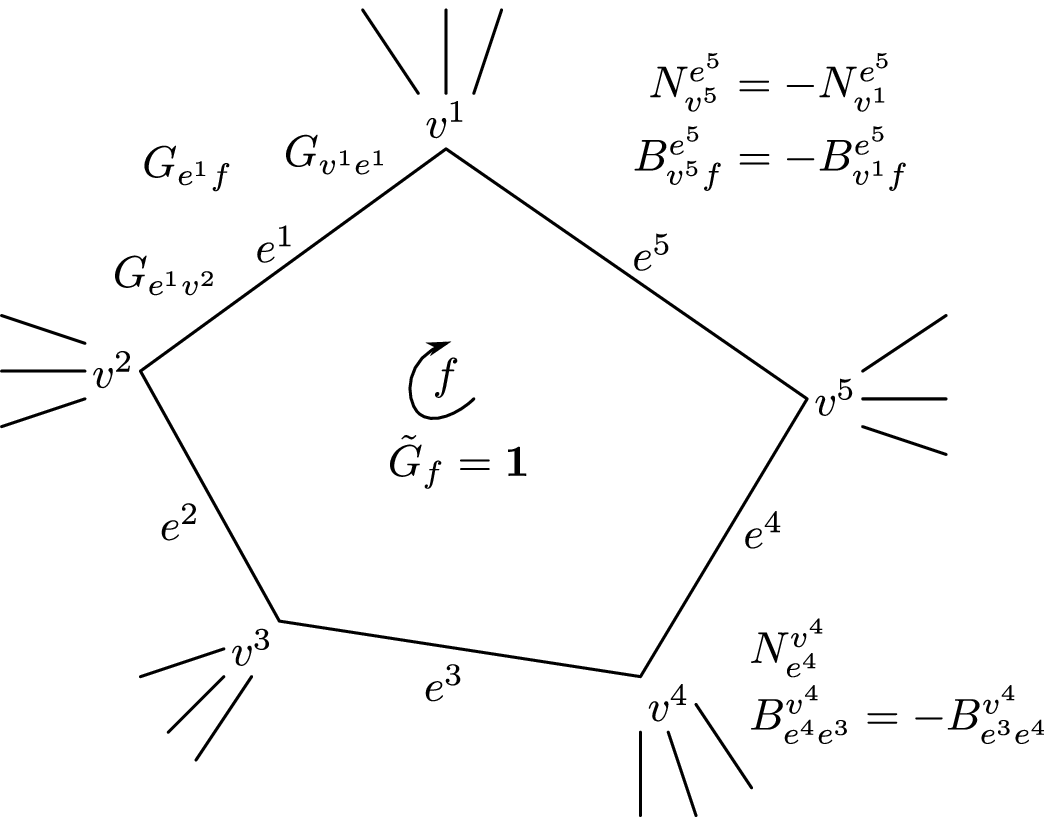}
 \caption{The data around a face $f=(v^1,e^1,v^2,e^2,v^3,e^3,v^4,e^4,v^5,e^5,v^1)$.}
 \label{fig-data1}
\end{figure}

\paragraph{Bivectors:} From the geometric data we can further construct bivectors, or $\spin(4)$ elements, associated to the triangles. This is easiest using the area normals introduced above:

\begin{align}\label{eq:B-def}
B^v_{ee'} &= \hodge N^v_e \wedge A^v_{ee'},\nn\\
B^e_{fv} &= \hodge N^e_v \wedge A^e_f.
\end{align}

Here $\hodge$ is the hodge dual with $\hodge^2 = \id$. By construction these bivectors lie in the plane of the triangles $t^v_{ee'}$ and $t^v_f$ respectively, are oriented,
\begin{align}\label{eq:B-orient}
B^v_{ee'} &= - B^v_{e'e},\nn\\
B^e_{fv} &= - B^e_{fv'}.
\end{align}
and satisfy the parallel transport equation,
\begin{equation}\label{eq:B-parallel transport}
G_{ev} \act B^v_{ee'} = B^e_{fv}.
\end{equation}
It is also immediate to see that they satisfy simplicity,
\begin{equation}\label{eq:B-simplicity}
N^e_v \cdot B^e_{fv} = \hodge  N^e_v \wedge N^e_v \wedge A^e_f = 0,
\end{equation}
due to the antisymmetrization in the $\hodge$, and closure
\begin{equation}\label{eq:B-closure}
\sum_{f \ni e} B^e_{fv} = \hodge N^e_v \wedge \left(\sum_{f \ni e} A^e_f\right) = 0,
\end{equation}
due to the closure of the area outward normals of $\tau^e$. The arrangement of all these data around a face is illustrated in Figure \ref{fig-data1}.

Further as $G_{ef}$ stabilizes the triangle $t^e_f$ it is generated by a bivector orthogonal to the plane of $t^e_f$, which is thus proportional to $\hodge \hat B^e_{fv}$, where $\hat B$ denotes the normalised bivector. That is,
\begin{equation}
G_{ef} = \exp_{\SO(4)} (\xi_{ef} \hodge \hat B^e_{fv}),
\end{equation}
and a straightforward calculation shows that, as it has to be, $\tilde G_f$ can be written as
\begin{equation}\label{eq:G_f=id, interior and deficit angle}
\tilde G_f = \exp_{\SO(4)} (\Theta_f + \sum_{e \in f} \xi_{ef}) \hodge \hat B^v_{ee^{n}}.
\end{equation}

This elucidates the geometric meaning of equation \eqref{eq:G_f=id}. It expresses the matching of the deficit angle around a face and the sum of the interior dihedral angles.

\paragraph{Twisted Bivectors:} Finally, anticipating the presence of the Immirzi parameter in the wave front set, we can introduce the twisted bivectors and holonomies, given in terms of the rational number $\gamma$. We will twist the geometric data by using again the twisting operator $T_{\gamma} = 1 + \gamma \hodge$, acting on bivectors. Note that as, as pointed out above, $T_{\gamma} T_{-\gamma} = (1-\gamma^2) \id$, this operator is invertible for $\gamma \neq 1$ which is always assumed in these spin foam models. Thus acting with this operator does not change the content of the set, merely parametrizes it differently. The twisted holonomies and bivectors are then given by

\begin{align}\label{eq:BgGg-def}
\Bg^{e}_{fv}&=T_\gamma B^{e}_{fv},\nn\\
\Bg^{v}_{ee'}&=T_\gamma B^{v}_{ee'},\nn\\
\Gg_{ev}&=G_{ev},\nn\\
\Gg_{ef}&=\exp_{\SO(4)}{\xi_{ef} \hodge {T_{-\gamma}} \hat B^{e}_{fv}} = 
{\exp_{\SO(4)}{\xi_{ef} \hodge R_\gamma{}^\gamma\hat{B}^{e}_{fv}}.}
\end{align}
{Where we again used the shortcut (see \eqref{eq:Rgamma})
\begin{equation}
 R_\gamma p=\frac{1}{\sqrt{1+\gamma^2}} T_{-\gamma} T_\gamma^{-1} p\ .
\end{equation}}

From these we can define $\Ggf$ as before, ${G}_f$ does not involve twisted variables and remains unchanged.

As $T_{\gamma}$ commutes with the adjoint action of $\Spin(4)$ on bivectors, the stability properties of bivectors remain unchanged, and so do all equations that can be derived from group elements and their action on bivectors. That is, only the equations involving normals change. These are in particular simplicity \eqref{eq:B-simplicity} and \eqref{eq:G_f=id}:

\begin{align}
N^e_v \cdot B^e_{vf} &= 0,\nn\\
\tilde G_f &= \id.\nn
\end{align}

The first of these simply turns into a twisted form of simplicity:

\begin{equation}\label{eq:TwistedSimplicity}
N^e_v \cdot T_{-\gamma} \Bg^{e}_{vf} = 0.
\end{equation}

The second of these is slightly more complicated. Instead of just providing the interior dihedral angles, the $\Gg_{ef}$ now also contain a rotation in the plane of the triangle. We thus obtain the twisted face equation

\begin{equation}\label{eq:G_f=id,twisted}
\tilde G_f = \exp_{\SO(4)} \left[ \bigg(\Theta_f + \sum_{e \in f} \xi_{ef}\bigg) \hodge \hat B^v_{ee^{n}} - \gamma \sum_{e \in f} \xi_{ef} \hat B^v_{ee^{n}}\right] = \exp_{\SO(4)} \bigg({\gamma \sum_{e \in f} \xi_{ef} \hat B^v_{ee^{n}}}\bigg).
\end{equation}

\subsection{Geometry from Holonomies and Bivectors}\label{sec:Geom from Hol and Biv}

After giving the holonomy and bivector variables in terms of the geometry, and deriving a number of equations automatically satisfied by them, the question arises whether these equations are not just necessary for a set of bivectors and holonomies to arise from a geometry in the above sense, but also sufficient.

We begin by giving the precise set of equations we require our bivectors and holonomies to satisfy. First at the vertex we have from \eqref{eq:B-orient}, \eqref{eq:B-parallel transport}, \eqref{eq:B-closure}, and \eqref{eq:TwistedSimplicity}:

\begin{align}\label{eq:GeomEqs1}
\Bg^{v}_{ee'} &= - \Bg^v_{e'e},\nn\\
\Bg^e_{vf} &= - \Bg^e_{v'f}, (v,v' \in e)\nn\\
\Bg^v_{ee'} &= G_{ve} \act \Bg^e_{vf}, \nn\\
\exists N^v_e: N^v_e \cdot T_{-\gamma}\Bg^v_{ef} &= 0,\nn\\
\displaystyle \sum_{\substack{e' \ni v\\e' \neq e}} \Bg^v_{ee'} &= 0.
\end{align}

Using the second equation, the latter two equations can equivalently be written on the edges as

\begin{align}\label{eq:GeomEqs2}
\exists N^e_v: N^e_v \cdot T_{-\gamma}\Bg^e_{vf} &= 0,\nn\\
\displaystyle \sum_{f \supset e} \Bg^e_{vf} &= 0,
\end{align}
where we also have from \eqref{eq:BgGg-def}, that
\begin{align}\label{eq:GeomEqs3}
\exists \xi_{ef}:\Gg_{ef}=\exp_{\SO(4)}{\xi_{ef} \hodge T_{-\gamma} \hat B^{e}_{fv}}.
\end{align}

Finally, on the faces we have from \eqref{eq:tildeG_f-stabilityTriag} and \eqref{eq:G_f=id,twisted} that
\begin{align}\label{eq:GeomEqs4}
\Ggf \act \Bg^v_{e'e} &= \Bg^v_{e'e},\nn\\
{G}_f \act T_{-\gamma}\Bg^v_{e'e} &= T_{-\gamma}\Bg^v_{e'e}.\nn\\
\Ggf &= \exp_{\SO(4)}{{-\gamma \sum \xi_{ef} \hat{B}^v_{ee'} }}. 
\end{align}
Here the $\xi_{ef}$ have to be the same as in \eqref{eq:GeomEqs3}, and $v$ is the fiducial vertex of the face $f$. The first two lines are actually implied by the parallel transport conditions and \eqref{eq:GeomEqs3}. 

Any bivectors and group elements constructed from a geometry in the sense of the previous subsection satisfy these equations. However, they are not the only solutions to these equations. We denote the set of solutions $\Bg$, $\Gg$ to the above equations $\GG^\gamma_\CC$. The classification of solutions follows directly from the reconstruction theorems in the literature. Before we review these, though, we first need to give a set of symmetries that act on the space $\GG^\gamma_\CC$.

There are two kinds of such symmetries, the first kind derives from the geometric origin of the equation \eqref{eq:GeomEqs1} to \eqref{eq:GeomEqs4}, the second preserve only the equations on group elements and bivectors, but change their relationship to the geometric quantities.

\paragraph{Geometric:} We can act with $G_v$, $G_e$ in $SO(4)$ on the geometric 4-simplices and tetrahedra located at the vertices and edges. This induces a transformation on the $\Bg$ and $\Gg$ that preserves geometricity and thereby the above equations. Explicitly it acts as:

\begin{align}\label{eq:GeomSymmetries}
\sigma^v &\rightarrow G_v \sigma^v & \tau^e &\rightarrow G_e \tau^e \nn\\
N^v_e &\rightarrow G_v N^v_e & N^e_v &\rightarrow G_e N^e_v \nn\\
B^v_{ee'} &\rightarrow G_v \act B^v_{ee'} & B^e_{vf} &\rightarrow G_e \act B^e_{vf} \nn\\
G_{ev} &\rightarrow G_e G_{ev} G_v^{-1} & G_{ef} &\rightarrow G_e G_{ef} G_e^{-1}
\end{align}

In the rest of the paper we will fix the $N^e_{v}$ to be proportional to the universal reference vector $N^e_v = \pm N^0 = \pm(1,0,0,0)$. Then the $G_e$ take value in the $\SO(3)$ subgroup which preserves this reference vector. This can be seen as a partial gauge fixing.

Further there is a set of symmetries which leaves the geometric content of the configuration untouched but locally changes the orientation. Note that acting with an $\OO(4)$ element $P_v$ that is not in $\SO(4)$ on $\sigma^v$ we retain the same geometry. Thus it is still possible to map the tetrahedra $\tau^e$ to the boundary of $P_v \sigma^v$ by an $\SO(4)$ element, however, the normal $N^e_v$ will now not be taken to the outward normal $N^v_e$, but the inward normal $-N^v_e$. Thus the bivectors get transported not to the geometric bivectors of $P_v \sigma^v$, but to the negative geometric bivectors. Calling the reflection with respect to the plane of $\tau^e$, $P_e$, this symmetry thus acts as \begin{align}
G_{ve} &\rightarrow P_v G_{ve} P_e, \nn\\
B^v_{ee'} &\rightarrow P^{-1}_v \act B^v_{ee'}.
\end{align}
The resulting bivectors do not have the same relationship to the geometry as before but are instead related to the geometric ones by a sign. The issue of these orientations is discussed in detail in the literature \cite{Barrett:2009cj,Barrett:2009gg,Hellmann:2010nf}.

Finally there is a global scaling symmetry that sends all $B$ to $\lambda B$. This will be important to understand when the partition function suffers from existence failure. We will come back to this issue in \ref{sec:cond-fin}.

\paragraph{Nongeometric:} The non-geometric symmetry affects the $G_{ef}$ and is novel to the setup we consider here. It acts as

\begin{equation} \Gg_{ef}  \rightarrow \Gg_{ef} \exp_{\SO(4)}(\xi'_{ef} {\hodge T_{-\gamma} }\hat{B}^e_{vf}), \label{eq:Symmetry-U(1)}\end{equation}
with $\sum_{e \in f} \xi'_{ef} = 2 \pi \N$. This transformation of course destroys the geometric interpretation of the individual $\Gg_{ef}$.

\paragraph{Reconstruction:} We will now use the reconstructions theorems from the literature to classify the elements in $\GG^\gamma_\CC$.

\paragraph{At the vertex:} Let us first consider the equations $\eqref{eq:GeomEqs1}$. These equations, and the extend to which the characterise a geometry, have received considerable attention in the literature \cite{Barrett:1993db,Barrett:1998gs,Conrady:2008ea,Conrady:2008mk,Barrett:2009cj,Barrett:2009gg,Barrett:2009as,Hellmann:2010nf,Barrett:2010ex,Engle:2011ps,Han:2011re,Han:2011rf}.

By theorems proved a long time ago in a galaxy far far away \cite{Barrett:1998gs,Conrady:2008ea,Barrett:2009cj,Hellmann:2010nf}, $\eqref{eq:GeomEqs1}$ suffices to reconstruct a geometry and orientation at every vertex if there are two additional non-degeneracy assumptions satisfied.

These non-degeneracy assumptions are the bivectorial version of the non-degeneracy of the geometry we assumed at the outset:
\begin{itemize}
 \item[3d:] The $A^e_f$ for fixed $e$ span a full 3-dimensional space. Equivalently, every set of three $B^e_{vf}$ for fixed $e$ and $v$ is linearly independent.
 \item[4d:] The $N^v_e$ span full 4-dimensional space. Equivalently, for every edge $e''$ at the vertex $v$ the six bivectors $B^v_{ee'}$ with $e$, $e' \neq e''$ span the full space of bivectors.
\end{itemize}

The failure of the second condition is well understood, and leads to a set of $\SU(2)$ $BF$ solutions \cite{Barrett:2009as}. At fixed bivector length this sector is five dimensions larger than the geometric sector and dominates in the Barrett-Crane vertex amplitude asymptotics \cite{Baez:2002rx}.

Further clearly a failure of the 3d condition implies a failure of the 4d condition. Thus the elements of $\GG^\gamma_\CC$ can be classified according to whether, at each vertex, the non-degeneracy assumptions are satisfied not at all, at the 3d level, or at the 4d level. In the latter case we can further differentiate between the two possible orientations. That is every solution defines a function $s(v)$, assigning every vertex either $+$, $-$, $\SU(2)$, or deg splits the solution space. The only fully oriented geometric solutions are when $s(v)=+$ for all $v$ or $s(v) = -$. Note that the orientation changing symmetry discussed above, acting at a vertex $v$, changes the $s(v) = +$ to $s(v) = -$ when in the geometric sector.

\paragraph{On the face:} The equations \eqref{eq:GeomEqs1} determine the geometry at the vertex, and thus already fix $\Gg_{ev} = G_{ev}$ to be a geometric connection. The geometric meaning of the $\Gg_{ef}$ is fixed by equation \eqref{eq:GeomEqs3}, and the last line of \eqref{eq:GeomEqs4}. Equation \eqref{eq:GeomEqs3} implies that the $\Gg_{ef}$ are fixed by the geometricity at the vertex up to their angle of rotation $\xi_{ef}$. To simplify the discussion we can now use the geometric symmetries to put us into a frame where all bivectors at a face equal some fiducial bivector, \[B^f = B^v_{ee'} = B^e_{vf}, \mbox{with} \, (e,v,e') \subset f.\] We then have

\begin{equation}\label{eq:Twisted-interwoven-face-holonomy}
\Ggf = \exp_{\SO(4)} \left[ \bigg(\Theta_f + \sum_{e \in f} \xi_{ef}\bigg) \hodge \hat B_f - \gamma \sum_{e \in f} \xi_{ef} \hat B_f\right]
\end{equation}

The last line of \eqref{eq:GeomEqs4}, which requires that \[\Ggf = \exp {-\gamma \sum_{e \in f} \xi_{ef} \hat B_f},\] now fixes the two angles of rotation:

\begin{align}\label{eq:FaceAngles}
\Theta_f + \sum_{e \in f} \xi_{ef} &= 0 \mod 2\pi\nn\\
\gamma \sum_{e \in f} \xi_{ef} &= \gamma \sum_{e \in f} \xi_{ef} \mod 2\pi
\end{align}

These equations, being invariant under the geometric symmetries, apply in all frames. Now using the symmetry \eqref{eq:Symmetry-U(1)}, we can indeed fix the $\xi_{ef}$ to be the interior dihedral angles, such that $\Gg_{ef}$ are just the geometric ones defined by \eqref{eq:BgGg-def}.

Thus, if for an element of $\GG^\gamma_\CC$, all vertices at a face $f$ are of the type $s(v)=+$ or are all of the type $s(v)=-$, the $\Bg$ at the vertices around that face, and the $\Gg_{ev}$, and $\Gg_{ef}$ around that face are the geometric ones up to the symmetry \eqref{eq:Symmetry-U(1)}.

\section{Geometry of the EPRL/FK wave front set: Accidental curvature constraints}\label{sec:Geom of PartFunc}

We can now match the set $\GG^\gamma_\CC$ to the interior closed set $\WW(\CC,\Gamma)$ arising in the wave front set analysis. Let us first recall the sets $\WW$ in question:

$\WW$ is given by the set of variables
\begin{equation}
 g_{ef},g_{ev}, p^{e}_{fv}, p^v_{ee'}
\end{equation}
satisfying
\begin{itemize}
 \item orientation and parallel transport equations:
\begin{align}
p^v_{e'e} = -p^v_{ee'} ,&\quad p^{e}_{fv} = - p^{e}_{fv'} \nn\\
p^{v}_{e'e} =  g_{ve} p^{e}_{fv}\ .&
\end{align}
\item closure: for all $e \notin \Gamma$, and $v \in e$ we have $\sum_{f\ni\ e} p^{e}_{fv}=0$.
\item holonomy: $\tilde{g}_f=1$ unless $p^{e}_{fv}=0$, otherwise no restriction,
\item simplicity: $(g_{ef},p^e_{fv})\in \WF(E)$,
\end{itemize}

We have two different type of simplicity constraints, the Barrett-Crane type and the FK/EPRL type. These are respectively

\begin{itemize}
\item (BC) $(g_{ef},p^e_{fv})$ such that $N^0 \cdot \hodge p^e_{fv} = 0$, $g_{ef} \act N^0 = N^0$.
\item {($\overline{\text{FK}}$) $(g_{ef},p^e_{fv})$ such that $N^0 \cdot \hodge p^e_{fv} = 0$, there exists a $\xi_{ef}$ such that $g_{ef}=e^{\xi\hodge \hat{p}^e_{fv}}$.}
\item (EPRL/FK$_{\gamma}$) $(g_{ef},p^e_{fv})$ such that ${N^0 \cdot T_{-\gamma} p^e_{fv}} = 0$, there exist a $\xi_{ef}$ such that {$g_{ef} = \exp(\xi_{ef} \hodge R_\gamma\hat{p}^e_{fv})$}
\end{itemize}

\subsection{{FK and EPRL for $\gamma=0$}}

We will begin with the simplest case, that of $FK = 0$. In this case we see that $\WW$ exactly covers the geometric set $\GG$. That is, the variables
\begin{equation}
(G_{ef}, G_{ev}, B^e_{vf}, B^v_{ee'})=(\hat{g}_{ef},\hat{g}_{ev},p^e_{vf},p^v_{ee'}),
\end{equation}
where $\hat{g}$ is the $\SO(4)$ element associated to the $\Spin(4)$ element $g$, satisfy all equations of $\GG$. Recall that the set $\GG$ decomposes according to the function $s(v)$, which associates to each vertex the level of degeneracy it satisfies. The same decomposition then applies in $\WW$. We have a further ambiguity in that for every element $\ww$ in $\WW$ that covers an element $\hat{\ww}$ in $\GG$ we have that for functions $l_{ev}: \CC_{ev} \rightarrow \pm$ and $l_{ef}: \CC_{ev} \rightarrow \pm$ satisfying $\prod_{(ev) \in f}l_{ev}\prod_{e \in f}l_{ef} = +$, the element $\ww'$ of $\WW$ given by changing $g_{ev} \rightarrow l_{ev} g_{ev}$, $g_{ef} \rightarrow l_{ef} g_{ef}$ cover the same element $\hat{\ww}$ of $\GG$. Thus we see that the association between elements of $\WW$ and $\GG$ is not injective.

Conversely it is always possible to find a set of $g_{ev}$, $g_{ef}$ covering a given element $\hat{g}$, and such that the equation $\tilde{g}_f = 1$ is satisfied by adjusting the coverings of $g_{ef}$

Thus we see that for $FK_{0}$, the set $\WW{\CC}$ splits according to the function $s(v)$ giving the degeneracy at the vertex, and the functions $l_{ev}$, $l_{ef}$ act freely on each of these layers.

Note that this argument use extensively the property of the $\WF(E^0_{\rm FK})$ that
\begin{equation}\label{eq:gef-lift}
 (g_{ef},p_{ef})\in \WF(E^0_{\rm FK})\Rightarrow (-g_{ef},p_{ef})\in \WF(E^0_{\rm FK})\ .
\end{equation}

\subsection{FK$_\gamma$, EPRL$_\gamma$}

Whereas in the previous case we had perfect matching between the geometric equations and the wave front set equations, this will break down here. The elements of a covered solution $\hat{\ww}$, satisfy all the equations of $\GG^\gamma$ except for \eqref{eq:G_f=id,twisted}:
\[
\Ggf = \exp_{\SO(4)} \bigg({-\gamma} \sum_{e \in f} \xi_{ef} \hat B^v_{ee^{n}}\bigg).
\]

Instead we have $\Ggf = \id$, as for the case $\GG^0$. That means, that instead of equation \eqref{eq:FaceAngles}, we have,
\begin{align}
\Theta_f + \sum_{e \in f} \xi_{ef} &= 0 \mod 2\pi\nn\\
\gamma \sum_{e \in f} \xi_{ef} &= 0 \mod 2\pi\ .
\end{align}
In \eqref{eq:FaceAngles}, and in the case $\gamma = 0$, the second equation is tautological. Here it is revealed as the origin of the so called flatness problem first pointed out by Bonzom.

Note that all other equations are unchanged, that is, elements of $\WW^\gamma$ still cover elements of $\GG^\gamma$. Thus in particular, the reconstruction theorems apply, and $\WW$ splits according to $s(v)$ again. Given a face $f$ such that $s(v) = +$ for all $v \in f$, that is, given a face in the geometric sector, these equations imply that the dihedral angle, $\Theta_f$, satisfies
\begin{equation}
\gamma \Theta_f = 0 \mod 2 \pi\ .
\end{equation}

Note that, as opposed to the case before, due to the presence of $\gamma$ the wave front set in general no longer satisfies the property \eqref{eq:gef-lift}\footnote{{
This property fails for example for $\gamma=\frac{1}{3}$, but if $\gamma$ is of the form $\frac{p}{q}$ where only either $p$ or $q$ is divisible by $2$ the condition \ref{eq:gef-lift} stil holds.}}.
It is no longer clear that every geometric configuration can be lifted. Thus it is in principle possible that despite the existence of compatible geometric configuration, the amplitude is suppressed due to this kind of spin structure obstruction.

\subsection{BC and $\overline{\text{FK}}$}
The BC and $\overline{\text{FK}}$ models structurally resemble the limit $\gamma \rightarrow \infty$. This is reflected in their geometry.

For the Barrett-Crane model the elements of $\WW$ no longer necessarily cover elements of $\GG$. The vertex reconstruction still proceeds as before, however, as $G_{ef} B^{e}_{fv} \neq B^{e}_{fv'}$, and the different $G_{ef}$ at an edge are uncorrelated, the geometry of the tetrahedron described by the bivectors $B^e_{fv}$ is changed.

That is, for BC the parallel transport condition $G_{ve}G_{ev'} \tau^{v'}_e = \tau^v_e$ fails. The geometry of neighbouring 4-simplices becomes uncorrelated.

The $\overline{\text{FK}}$ model differs from the BC model in this point, there the parallel transport condition $G_{ef} B^{e}_{fv} = B^{e}_{fv'}$ holds. The reconstruction proceed thus as in FK$_\gamma$ and EPRL$_\gamma$ case.

However for both models the flatness problem remains in its extreme form. We have simply that
\begin{equation}
\gamma \Theta_f = 0 \mod 2 \pi\ .
\end{equation}
For these models the $SO(4)$ holonomy in the geometric sector must be the identity.

\subsection{A corrected partition function}

As we saw, for the cases involving $\gamma$ the wave front set includes a spurious equation that restricted the curvature in the geometric sector to a discrete set. This arises as the face distribution, $\delta(\tilde g_{f})$ is not twisted along with the simplicity constraints $E(g_{ef})$. While the $g_{ef}$ then include a torsion contribution, the $g_{ev}$ remain the torsion free geometric connection.

In order to fix this issue we can twist the face amplitude. The role of the face amplitude in the wave front set is to enforce parallel transport, and equation \eqref{eq:FaceAngles} that enforces \begin{equation}\Theta_f + \sum_{e \in f} \xi_{ef} = 0 \text{ mod } 2 \pi\ ,\end{equation} or more geometrically
\begin{equation}\sum_{e \in f} \xi_{ef} = 2 \pi - \Theta_f \text{ mod } 2 \pi\ ,\end{equation}
which allows us to interpret the $\xi_{ef}$ as the interior dihedral rotations.

The new face amplitude $\omega'(g_{ev},g_{ef},\ldots)$ should enforce these two conditions. Recall that the parallel transport condition arose purely from the multiplicative structure of the functional dependence of $\omega$. If $\omega'$ is of the form
\begin{equation}\omega'(g_{ev},g_{ef},\ldots) = D(g_f,\tilde g_f), \end{equation} we automatically have guaranteed that the parallel transport conditions are still satisfied.

The wave front set of $\omega'$ given in terms of
\begin{equation}
 \{g_{ef},g_{ev},\ldots,p_{ef},p_{ev},\ldots\}
\end{equation}
can be written in terms of auxiliary variables by application of lemma \ref{lem:extension and integration}
properties \ref{pro-delta}, \ref{pro-parallel} and \ref{pro-inverse} . These auxiliary variables are
 \begin{align}
\tilde{q}_{ef} =  g_{ev}\act \tilde{q}_{ev} ,&\quad \tilde{q}_{ev} =  g_{ve'}\act \tilde{q}_{ve'},\nn\\
\tilde{q}_{ve'} = g_{e'f}\act \tilde{q}_{e'f} ,&\quad \tilde{q}_{e'f} =  {g_{e'v'}}\act \tilde{q}_{e'v'}.\\
{q}_{ev}=g_{ve'}\act{q}_{ve'},&\quad {q}_{ve'}={g_{e'v'}}\act{q}_{e'v'}
\end{align}
such that
\begin{equation}
 (g,\tilde{g},q_{ev},\tilde{q}_{ev})\in \WF(D)
\end{equation}
for $v$ the fiducial vertex of the face $f$. The wave front set is then given in terms of these auxiliary variables by
\begin{equation}
 p_{ef}=\tilde{q}_{ef},\quad p_{ev}=q_{ev}+\tilde{q}_{ev}\ .
\end{equation}
As, due to the wave front set of the $E$ function, $\tilde{q}_{ef}$ is stabilised by $g_{ef}$, the $p_{ef}$ and $p_{ev}$ also satisfy parallel transport.

We thus need to choose our function $D(g_{f},\tilde g_{f})$ such that it enforces only \eqref{eq:FaceAngles}. Note that as we have all equations up to \eqref{eq:FaceAngles} the full geometric reconstruction can already be applied, with the exception of the interpretation of the $\xi_{ef}$.

Without known anything further about $D$ we can thus already conclude (see equation \eqref{eq:Twisted-interwoven-face-holonomy} for example) that in the wave front set  $g_{f}$ and $\tilde g_{f}$ have to be of the forms
\[g_f = \pm \exp \Theta_f \hodge \hat B_f\]
and,
\[\tilde g_f = \pm \exp\left[ \bigg(\Theta_f + \sum_{e \in f} \xi_{ef}\bigg) \hodge \hat B_f - \gamma \sum_{e \in f} \xi_{ef} \hat B_f\right]\ .\]{where $B_f$ is a simple bivector obtained from  $p_{ef}$.
\begin{equation}
\hat{B}_f=\widehat{T_\gamma^{-1}p_{ef}}
\end{equation}

We also know that $\tilde q_{ev}$, being the parallel transport of the bivector in the wave front set of the e function has to be of the form $\tilde q \propto T_\gamma (\wec{n},\wec{n})$.

We can now try the ansatz 
\begin{align} 
D(g,\tilde{g}) = \sum_{k\in {\mathbb N}} \int_{S^2\times S^2}& \dd \wec{n} \,\dd \wec{m}\, \langle \epsilon_+ n|g^+|\epsilon_+ n \rangle^{|\alpha+\beta|k} \langle \epsilon_- m|g^-|\epsilon_- m \rangle^{|\alpha-\beta|k}\times\nn\\&\times  \langle \tilde \epsilon_+ n|\tilde{g}^+|\tilde \epsilon_+ n \rangle^{|\tilde\alpha+\tilde\beta|k} \langle \tilde \epsilon_- m|\tilde{g}^-|\tilde \epsilon_- m \rangle^{|\tilde\alpha-\tilde\beta|k}\ . 
\end{align}
where $\alpha,\beta, \tilde\alpha,\tilde\beta\in{\mathbb Z}$ are to be determined, and $\epsilon_{\pm}|\alpha \pm \beta| = \alpha \pm \beta$ and $\tilde \epsilon_{\pm}|\tilde \alpha \pm \tilde \beta| = \tilde \alpha \pm \tilde \beta$.

From the properties and Lemmas given, we can compute the wave front set of this distribution
\begin{equation}
\begin{split}
 \WF(D)\subset&\Big\{(g,\tilde{g},q,\tilde{q})\colon
q=\alpha B+\beta \hodge B, \tilde{q}=\tilde\alpha B+\tilde\beta \hodge B,
g=e^{\chi \hat B+\rho\hodge \hat B}, \tilde{g}=e^{\tilde\chi \hat B+\tilde\rho\hodge \hat B},\\
&\text{where } B=\xi(\hat{\wec{n}},\hat{\wec{m}}),
2\alpha\chi+2\beta\rho+2\tilde\alpha\tilde\chi+2\tilde\beta\tilde\rho
= 0 \text{ mod } 4\pi\Big\}
\end{split}
\end{equation}

First note that to have $\tilde q \propto T_\gamma (\wec{n},\wec{n})$, as required by the parallel transport condition, we have to have $\alpha = \gamma \beta \neq 0$.

Now in the wave front set we have from the form of $g_f$ and $\tilde g_f$ that
\begin{equation}
 \begin{array}{lcl}
 g_f = \pm e^{\Theta_f \hodge \hat B_f} & \Rightarrow &  \chi=0, \rho=\Theta_f \\
 \tilde g_f = \pm e^{\left(\Theta_f + \sum_{e \in f} \xi_{ef}\right) \hodge \hat B_f - \gamma \sum_{e \in f} \xi_{ef} \hat B_f} &\Rightarrow & \tilde\chi=\Theta_f + \sum_{e \in f} \xi_{ef},\ \tilde\rho=- \gamma \sum_{e \in f} \xi_{ef} \ .
 \end{array}
\end{equation}

The equalities on the right hand side are modulo $2\pi$. Thus the equation we want the wave front set to encode is simply that $\tilde \chi = 0$ while all other angles are free\footnote{We can't simply set $\tilde \alpha = 1$ and all other coefficients to zero, as we then could no longer satisfy the parallel transport conditions.}: \begin{equation}\label{eq:DfuncWFScondition}
 \alpha\chi+\beta\rho+\tilde\alpha\tilde\chi+\tilde\beta\tilde\rho
= 0 \text{ mod } 2\pi\Leftrightarrow \tilde\chi=0 \text{ mod } 2\pi\ .
\end{equation}

We now assume that $\alpha$ and $\beta$ are the minimal integers satisfying $\gamma=\frac{\beta}{\alpha}$, then we can eliminate $\Theta_f$ and $\sum_{e \in f} \xi_{ef}$ from the conditions derived from geometricity,
\begin{equation}
\rho=\Theta_f,\ \tilde\chi=\Theta_f + \sum_{e \in f} \xi_{ef}+2\pi m,\ \tilde\rho=- \gamma \sum_{e \in f} \xi_{ef}+2\pi m'
\end{equation}
to obtain relations between the $\chi$, $\tilde \chi$, $\rho$ and $\tilde \rho$
\begin{equation}
\exists m,m'\colon  \tilde\chi-\rho=-\frac{\tilde\rho}{\gamma}+\frac{2\pi \alpha m'}{\beta}+2\pi m 
\end{equation}

which is equivalent to

\begin{equation}
\beta(\tilde\chi-\rho)=-\alpha\tilde\rho \text{ mod } 2\pi
\end{equation}
because $\alpha$ and $\beta$ have no common divisors.

Using $\chi=0$, $\beta\rho=\gamma\alpha\tilde\chi+\alpha\tilde\rho$ mod $2\pi$ and $\beta=\gamma\alpha$ we can transform 
\begin{align}
 \alpha\chi+\beta\rho+\tilde\alpha\tilde\chi+\tilde\beta\tilde\rho
&= 0 \text{ mod } 2\pi\Leftrightarrow\nn\\
(\gamma \alpha + \tilde\alpha) \tilde \chi + (\alpha + \tilde\beta)\tilde\rho &= 0 \text{ mod } 2\pi 
\end{align}

Thus in order to have \eqref{eq:DfuncWFScondition} we have to have
\begin{equation}
 \alpha\gamma+\tilde\alpha=1 ,\quad \alpha+\tilde\beta=0 ,\quad  \beta=\gamma\alpha\ , 
\end{equation}
where $\alpha$ and $\beta$ are the minimal integer pair satisfying $\beta=\gamma\alpha$.
We can further compute that the wave front set bivectors of $\omega'$ are of the form
\begin{equation}
 p_{ev}=q_{ev}+\tilde{q}_{ev}=
 ((\alpha+\tilde{\alpha})\wec{n},(\beta+\tilde{\beta})\wec{n})=\alpha (1-\gamma)(\wec{n},-\wec{n})+(\wec{n},0)\ .
 \end{equation}
 
It should be noted that with this face amplitude, the connection to the $\SU(2)$ Hilbert space, sometimes touted as a desirable feature of the EPRL model, is lost. In this proposition we also depart from the general form of the partition function for holonomy spin foam models. {The question remains if one also obtains a connection to the Regge action in the large $j$ limit with this face amplitude. This issue, however, cannot be address by our method.}

\section{Further properties and an example}\label{sec:FurtherProperties}

In this section we will show that the wave front set analysis given above extends exactly to the type of face weights used in the literature. We then show that the wave front set analysis allows us to give exactly the large spin asymptotics of the partition function. Finally we will give a concrete example in the form of the 3-3 move, where our conditional is satisfied.

\subsection{General face weights}
\label{sec:face-wf}
In the literature, we find a wide rang elf proposals for the face amplitude of spin foam models, which all lead to different face distributions $\omega$. As the theories are motivated from BF theory with $\omega = \delta$, we have so far assumed that the wave front set properties of $\omega$ and those of $\delta$ coincide. This is indeed the case for a wide class of face amplitudes, including most that have been considered in the literature.

Specifically consider the general form of face amplitude defined by
\begin{equation}
 \omega(g)=\sum_{j_+,j_-} \tilde{d}_{j_+,j_-}
\int_{S^2\times S^2} \dd \wec{n}_+\ \dd\wec{n}_-\ \langle n_+|g^+|n_+\rangle^{2j^+}\langle n_-|g^-|n_-\rangle^{2j^-}
\end{equation}
where $\tilde{d}_{j_+,j_-}=d(x,y)$ is the face amplitude. 
Then assume that $d(x,y)$  satisfies (see \ref{sec:ellipticity} for definitions)
\begin{equation}
 d\in S^\rho({\mathbb R}^2),\quad d^{-1}\in S^{-\rho}({\mathbb R}_+^2)
\end{equation}
for some $\rho\in{\mathbb R}$. All proposed spin foam amplitudes satisfy this condition, for example
\begin{itemize}
 \item the $SO(4)$ face weight 
\begin{equation}
 d(x,y)=(2x+1)^2(2y+1)^2,
\end{equation}
and similarly all polynomial $d$ functions.
\item the $\SU(2)$ face weight 
\begin{equation}
 \frac{\sqrt{2}}{\sqrt{1+\gamma^2}}\sqrt{x^2+y^2}(2x+1)(2y+1)
\end{equation}
It can be checked that for $j^\pm=\frac{|1\pm\gamma|}{2}k$, 
\begin{equation}
 d(j^+,j^-)=k(2j^++1)(2j^-+1),
\end{equation}
\end{itemize}

We can then use lemma \ref{lem-FK}, and properties \ref{pro:ellip}, \ref{pro-int} (with the case of a unique projection of the interior closed set to the boundary) to
deduce that
\begin{equation}
 \WF(\omega)=\WF(\delta)=\{0,\mathfrak{g}\}\cup \{0\}
\end{equation}

One class of examples that does not satisfy these assumptions are the heat kernel regularised spin foam amplitudes. These can be seen to be smooth.

\subsection{The geometry of the existence condition of the partition function}
\label{sec:cond-fin}

The geometric meaning of the condition \ref{eq-exists2} is not clear. In this section we will relate it to the large bivector limit in the interior. In fact suppose that we have a sequence of configurations
\begin{equation}
 \{{}^lg_{ev},{}^lg_{ef},{}^lp^f_{ev},\cdots)
\end{equation}
labelled by integer $l$ with the property that
\begin{equation}
 \frac{\sum_{e\not\in \Gamma} |{}^lp^f_{ev}|}{\sum_{e\in \Gamma} |{}^lp^f_{ev}|}\rightarrow \infty
\end{equation}
that is exactly condition that the internal bivectors become much larger than the boundary ones. Let denote
\begin{equation}
 \lambda_l=\frac{1}{\sum_{e\in\CC} |{}^lp^f_{ev}|}
\end{equation}
then the rescaled configurations
\begin{equation}\label{eq-conf1}
 \{{}^lg_{ev},{}^lg_{ef},{}^l\tilde{p}^f_{ev}=\lambda_l {}^lp^f_{ev},\cdots)
\end{equation}
still belong to $\WW(\CC,\Gamma)$.
The space of all $\{g_{ev},\ldots , p^f_{ev}\}$ with the condition $\sum_{e\in\CC} |{}^lp^f_{ev}|=1$ is however compact, and $\WW(\CC,\Gamma)$ is closed. Thus there exists a convergent subsequence of \eqref{eq-conf1} to the point
\begin{equation}
 \{\tilde{g}_{ev},\dots,\tilde{p}^f_{ev}\}
\end{equation}
It is easy to check that on the boundary $\tilde{p}^f_{ev}=0$ whereas, as we have $\sum_{e\in\CC} |\tilde{p}^f_{ev}|=1$, there need to be nonzero bivectors in the interior.

This proves that if we can make interior bivectors arbitrary large while keeping the boundary $B$'s small, condition \eqref{eq-exists2} is not satisfied. Conversely if a finite interior with zero boundary exists, this can of course be scale up arbitrarily.

\subsection{Relation between large $j$ limit and the wave front set}\label{sec:large j limit}

So far we have derived our results on the behaviour of the partition function in terms of the wave front set of the resulting distribution. This raises the question how the regime that we probe is related to the regime most studied so far, that of the large spin asymptotics. With the results we have derived so far we can indeed show that these two regimes coincide, our wave front set analysis provides exactly the large spin asymptotics of the partition function.

To see this we can use the distributions described in Lemma \ref{lem-FK}. Given a set of of $\SU(2)$ coherent states on the boundary
\begin{equation}
|n^{v\pm}_{e}\rangle, |n^{e\pm}_{v}\rangle,\colon \{ev\}\in \Gamma
\end{equation}
we can introduce a sequence of coherent states in universal boundary Hilbert space
\begin{equation}
 \psi_\lambda(g_{ev},\ldots)=\prod_{ev\in \Gamma}
\langle n^{e+}_{v}|g^+_{ev}|n^{v+}_{e}\rangle^{2\lambda j^+_{ev}}
\langle n^{e-}_{v}|g^-_{ev}|n^{v-}_{e}\rangle^{2\lambda j^-_{ev}}\ .
\end{equation}
We are interested in determining the behaviour of
\begin{equation}
 \lambda\rightarrow \langle\overline{\psi_\lambda},\ZZ\rangle
\end{equation}
for large $\lambda$.

The coherent state kernel distribution $K_{n^e_{v}}(x,g_{ev},\ldots)$, integrated with
$e^{i\lambda x}$ provides exactly this coherent state:
\begin{equation}
 \psi_\lambda(g_{ev},\ldots)=\int \dd x\ e^{i\lambda x}K_{n^e_{v}}(x,g_{ev},\ldots),
\end{equation}
Let us contract $K_{n_{ev}}(x,\ldots)$ with the amplitude. The result is a distribution $P(x)$ on $S^1$: 
\begin{equation}
 P(x)=\int\prod \dd g_{ev}\ K_{n^e_{v}}(g_{ev},\ldots)\ZZ(g_{ev},\ldots)
\end{equation}
Note that
\begin{equation}
 \int \dd x\ P(x)e^{i\lambda x}=\langle \overline{\psi_\lambda},\ZZ\rangle\ ,
\end{equation}
thus, if $P(x)$ is a smooth function, that is, its wave front set vanishes, the amplitude for the probed coherent states vanishes like $O(\lambda^{-\infty})$ in the large $j$ limit. 

We thus obtain a necessary condition on the existence of a non-suppressed large $j$ limit from the wave front set of the distribution $P(x)$.  Denote by $p^{v}_e$ and $p^e_{v}$ the bivectors associated with $|n^{v\pm}_{e}\rangle$ and $|n^{e\pm}_{v}\rangle$,
\begin{equation}
 p^v_{e}=(2j^+\wec{n}^{v+}_{e},2j^-\wec{n}^{v-}_{e}),\quad p^e_{v}=(2j^+\wec{n}^{e+}_{v},2j^-\wec{n}^{e-}_v)\ .
\end{equation}
Now combining the wave front set of $K_{n^e_{v}}$ from lemma \ref{lem-FK} with that of $\ZZ$, using property \ref{pro-int}
we see that $P(x)$ is smooth (that is, the large $j$ behaviour is $O(\lambda^{-\infty})$) if
\begin{equation}
 \not\exists\ \{ g_{ev}\colon (ev)\in\Gamma\} \colon\
p^e_{v}=g_{ev}\act p^v_{e},\quad
(g_{ev},\ldots; - p^v_{e}\ldots)\in \WF(\ZZ)
\end{equation}

Thus the suppression of the amplitude is governed by the wave front set. Further the singular support of $P(x)$ gives us information on the phase of the partition function. Note that the wave front set contains potentially more information than the coherent states. The coherent states only probe the existence of a $g_{ev}$ satisfying the above, the wave front set can tell us which one actually occurs. Conversely, for every element of the wave front set there clearly is a coherent state probing this point.

Thus we have that coherent state boundary data is suppressed if the geometry it describes is not the boundary geometry of an internal solution. The relationship between the inner and the boundary geometry is that for a boundary vertex $v$ from which we have the unique internal edge $e'$ we have a tetrahedron $\tau^v$ that has as it's bivectors the $p^v_{ee'}$ in the wave front set as given by the boundary operator \eqref{eq:gev-bdry}. These tetrahedra have the same geometry as the boundary tetrahedron $\tau^{v'}_{e'}$, but are arbitrarily rotated.

A more detailed discussion of the boundary geometry and wave front set for the case of projected spin network boundaries can be found in appendix \ref{sec:ProjectedSpinNetworkBoundary}.

\input{AD-JHEP-33Move.tex}

\input{AD-Discussion.tex}

\acknowledgments
 We would like to thank Bianca Dittrich for fruitful discussions and comments.
WK acknowledges the grant of Polish Narodowe Centrum Nauki number 501/11-02-00/66-4162. Research
at Perimeter Institute is supported by the Government of Canada through Industry Canada and by the Province
of Ontario through the Ministry of Research and Innovation.

\bibliographystyle{JHEP}
\bibliography{AD-JHEP}
\appendix
\input{AD-JHEP-Appendix1.tex}
\input{AD-JHEP-Appendix2.tex}
\input{AD-JHEP-Appendix3.tex}

\section{Conventions}\label{sec:conv}
We briefly summarise our conventions for various structures around $\Spin(4)$, $\SO(4)$, $\SU(2)$, and $\SO(3)$. For elements of $\su(2)$ and their eigenstates, the coherent states, we write
\[
\wec{n} \in \su(2),\ \hat{\wec{n}} = \frac{\wec{n}}{|\wec{n}|},\ 
\hat{\wec{n}} | n \rangle = \frac{i}2 | n \rangle\ .
\] 
{Components of the vector will be denoted by} $n^i$.
For the exponential map in the group $\SU(2)$ (and its image in $\SO(3)$) we use the convention
\begin{equation*}
\begin{array}{lll}
 \exp (\wec{n}) \in \SU(2), & \exp (2 \pi\ \hat{\wec{n}}) = -\id \in \SU(2), & \exp (4 \pi\ \hat{\wec{n}}) = \id \in \SU(2)\\
 \exp_{\SO(3)} (\wec{n}) \in \SO(3), & \exp_{\SO(3)} (2\pi\ \hat{\wec{n}}) = \id \in \SO(3), &\\ 
\end{array}
\end{equation*}
For bivectors we have:
\begin{equation*}
\begin{array}{ll}
 B = (\wec{n},\wec{n}') \in \spin(4) = \so(4), & \hodge (\wec{n},\wec{n}') = (\wec{n},-\wec{n}') \\
 T_{\gamma} = 1 + \gamma \hodge, & R_{\gamma} = \frac1{\sqrt{1+\gamma^2}} T_{-\gamma} T_{\gamma}^{-1}\\
 |B|^2 = |(\wec{n},\wec{n}')|^2 = \frac12 (|\wec{n}|^2 + |\wec{n}'|^2),\ & 
 |T_{\gamma}(\wec{n},\wec{n})|^2 = 
  (1+\gamma^2)|\wec{n}|^2
\end{array} 
\end{equation*}
In the case of $\Spin(4)$ and $\SO(4)$ these are
\begin{equation*}
\begin{array}{ll}
 g = (g^+,g^-) \in \Spin(4),\ &  \exp (B)\in \Spin(4),\\
 G \in \SO(4),  & \exp_{\SO(4)} (B)\in \SO(4)
\end{array} 
\end{equation*}
{and we also use following convention for the action of $\SU(2)$ ($\SO(3)$) on vectors
\begin{equation}
 (h\akt \wec{n})^i\sigma_i= n^i\ h \sigma_i h^{-1}
\end{equation}
with $\sigma_i$ the Pauli matrices and similarly for action of $\Spin(4)$ ($\SO(4)$)
\begin{equation*}
 (\exp_{\SO(4)} (\wec{n},\wec{n}') \act N)^I\sigma_I = N^I \exp(\wec{n}) \sigma_I \exp(\wec{n}')^{-1},
\end{equation*}}
with $\sigma_0 = \id$. In particular the subgroup generated by bivectors of the form $(\wec{n},\wec{n})$ stabilises the north pole $N^0 =\delta^{0I} = (1,0,0,0)$. Thus as bivectors we have $N^0 \cdot B = 0$ if $B = (\wec{n},\wec{n})$.

\end{document}

%% file: Macros.tex
\newcommand{\bea}{\begin{eqnarray}}
\newcommand{\eea}{\end{eqnarray}}
\newcommand{\nn}{\nonumber}
\newcommand{\be}{\begin{equation}}
\newcommand{\ee}{\end{equation}}

\theoremstyle{definition}
\newtheorem{defi}{Definition}
\newtheorem{pro}[defi]{Property}
\theoremstyle{plain}
\newtheorem{theo}[defi]{Theorem}
\newtheorem{lem}[defi]{Lemma}
\newtheorem{prop}[defi]{Proposition}

\newcommand{\Diag}{\mathrm{Diag}}

\newcommand{\N}{\mathbb{N}}

\newcommand{\Zi}{\mathcal{Z}_{\mbox{int}}}
\newcommand{\CC}{\mathcal{C}}
\newcommand{\DD}{\mathcal{D}}
\newcommand{\GG}{\mathcal{G}}
\newcommand{\WW}{\mathcal{W}}

\newcommand{\R}{\mathbb{R}}

\newcommand{\ww}{{\bf w}}

\newcommand{\calD}{{\mathcal D}}

\newcommand{\id}{\mathbf{1}}
\DeclareMathOperator{\tr}{tr}

\newcommand{\ZZ}{\mathcal{Z}}

\newcommand{\act}{\triangleright}

\newcommand{\hodge}{\ast}

\newcommand{\lalg}[1]{\mathfrak{#1}}  
\newcommand{\SU}{\mathrm{SU}}
\newcommand{\SO}{\mathrm{SO}}

\newcommand{\OO}{\mathrm{O}}

\newcommand{\Spin}{\mathrm{Spin}}

\newcommand{\su}{\lalg{su}}

\newcommand{\so}{\lalg{so}}

\newcommand{\spin}{\lalg{spin}}
\newcommand{\dd}{\mathrm{d}}

\def\la{\langle}
\def\ra{\rangle}

%% file: AD-JHEP-Summary.tex

\subsection{Summary and results}

The main result of this paper is a sufficient condition for the partition function to exist as a distribution, and a necessary condition for the boundary data to be not suppressed exponentially.

Both of these conditions are formulated in terms of a set of bivectors and group elements on the interior of the 2-complex. These satisfy a set of equations and continue the boundary variables into the interior. The sufficient condition for the partition function to exist is that the bivectors on the interior can not get arbitrarily large without the bivectors on the boundary growing likewise. A necessary condition for the boundary data to not be suppressed is that it can be continued into a solution on the interior.

The solutions to the interior equations can be classified vertex by vertex using the classification theorems of \cite{Barrett:2009gg,Barrett:2010ex,Barrett:2009mw,Barrett:2009as,Hellmann:2010nf} into degenerate, classical $\SU(2)$ BF and geometric solutions, where the latter come in two orientations. We capture this formally, by defining a function $s(v)$ for each solution, $s(v): \CC_v \rightarrow \{\text{deg,BF,}+,-\}$ that takes the value deg for fully degenerate vertices, BF for the $\SU(2)$ solutions and $+$ and $-$ for the positively and negatively oriented solutions respectively.

Beyond the geometricity conditions the solutions on the interior also satisfy, face by face, the equation 
\begin{equation}\label{eq:AccCurvConst}
\gamma \Theta_f = 0 \mod 2 \pi\ ,
\end{equation}
where $\gamma$ is the Immirzi parameter, and $\Theta_f$ is an angle derived from the geometry of the vertices $v$ around the face $f$. If $s(v) = +$ or $s(v) = -$, that is, the interior solution is entirely geometric and consistently oriented around the face in question, $\Theta_f$ is the deficit angle of the holonomy around the face. 

Thus we show that a boundary geometry that can only be continued into the interior with a geometry that contains curvature not satisfying \eqref{eq:AccCurvConst}, is exponentially suppressed unless $\gamma = 0$. This issue arises due to the fact that only some of the equations encoding the geometry in the bivector and holonomy language are changed upon the introduction of $\gamma$. We thus consider this equation an accidental curvature constraint.

This can be seen as a refinement of the flatness problem first pointed out by Bonzom in \cite{Bonzom:2009hw,Bonzom:2009wm}. This refined result, and the proof strategy were first announced in \cite{HellKami}, then heuristically discussed for the Lorentzian case by Perini \cite{Perini2012}, and then derived exactly for the Lorentzian case by different means by Han \cite{Han2013a,Han2013b}, where furthermore explicit bounds for the suppression away from these solutions are given.

Note that our partition function is a convolution of distributions. The main mathematical tool developed to study such a product of distributions is their wave front set. This can be seen as the subset of the cotangent bundle over the space on which the distribution is defined for which the distribution is not suppressed in the limit of large covectors. We introduce this notion in detail in section \ref{sec:WaveFrontSets}. Our first aim is to collect results from the literature and adapt them to our setting of distributions on group manifolds. In section \ref{sec:CompAndMult} we discuss the underlying theorems that allow us to state the sufficient and neccessary condition for existence and non-suppressed behaviour, which we then combine into lemma \ref{lem:extension and integration}. Section \ref{sec:GeneralPartFunc} then applies this to the case of our partition function, expressing the wave front set of $\ZZ$ in terms of the of $\omega$ and $E$. Section \ref{sec:Analysis on Group manifolds} then reviews and adapts the technical results needed to derive the wave front set of $\omega$. This concerns mostly distributions with arguments that are products of group elements. The wave front set of $\omega$ provides us with the parallel transport and orientation equations for our interior data. We then apply this to the case of the partition function in section \ref{sec:FaceDist}.

Thus we have the wave front set of the partition function in terms of the wave front set of the distribution $E$. Up to this point all spin foam models that can be cast into the form \eqref{eq:Hol-SFM-partfunc} with typical face weights coincide. The $E$ function encodes the choice of simplicity constraints, and the next two sections deal with deriving the properties of the various simplicity constraints in use.

In the cases of BF theory and the Barrett-Crane model, discussed in section \ref{sec:BFandBC} the $E$ function is simply a delta function making this analysis straightforward. We first give the asymptotic equations of BF theory and see that they match exactly the classical equations of BF. In section \ref{sec:BC} we derive the equations for the Barrett-Crane model.

For the EPRL and the Freidel-Krasnov models the $E$ function is much more complicated. The derivation of their wave front set is the subject of section \ref{sec:EPRL-FK E func}. For the EPRL $E$ function we obtain the wave front set by analysing a set of differential equations that are solved by the EPRL function. In section \ref{sec:Prop Actions and diff equations} we review and develop how such differential equations restrict the wave front set of the distributions that solve them, in section \ref{sec:EPRL function} we derive the wave front set of the EPRL function, with much of the technical discussion given in appendix \ref{sec-E-der}. For Freidel-Krasnov we require results on the asymptotics of distributions defined in terms of coherent states. The main technical result is lemma \ref{lem-FK} proven in section \ref{sec:coherent_state_dist_kernel}. This is used in sections \ref{sec:The FK model gamma larger 1} and \ref{sec:FK bar} to derive the wave front set of the FK $E$ function.

In section \ref{sec:EPRL-FK partition function wave front set} we combine the results on the $E$ functions with the results of the previous section to give a complete set of equations for the interior set of bivectors and group elements.

In order to understand the geometric meaning of these equations we next turn towards describing geometric configurations in these variables in section \ref{sec:Geometry}. We first describe what type of bivectors and group elements can be constructed naturally from a discrete geometry in section \ref{sec:Hol and Biv from Geo}, and note that the equations they satisfy mirror those obtained from the wave front set. We then show how these equations change if we introduce a twisting parameter $\gamma$ that turns simple bivectors into twisted simple bivectors.

In section \ref{sec:Geom from Hol and Biv} we then discuss how to invert the above construction by deriving the geometry from the equations. To do so we first discuss the symmetries of the equations given and then discuss how to apply the reconstruction theorems from the literature to classify the solution space according the the $s(v)$ introduced above.

In section \ref{sec:Geom of PartFunc} we then combine the results of section \ref{sec:Geometry} and section \ref{sec:EPRL-FK partition function wave front set} to give a description of the space of solutions to the interior equations derived in section \ref{sec:EPRL-FK partition function wave front set} in terms of their geometric content. We identify $\Theta_f$ as the deficit angle of the geometry on the interior, and discuss in detail where there are mismatches between the geometricity equations and the wave front set equations. The only case where there is no mismatch is that of $\gamma = 0$ EPRL and FK. For $\gamma \neq 0$ we have the accidental curvature constraint above, and for the $\gamma = \infty$ models, $\overline{\text{FK}}$ and Barrett-Crane, we obtain exact flatness. For Barrett and Crane we furthermore have that the boundary geometry of neighbouring 4-simplices can differ.

In section \ref{sec:FurtherProperties} we then complete the above analysis and give an example. In section \ref{sec:face-wf} we show that the results hold for a very wide array of face weights discussed in the literature. In section \ref{sec:cond-fin} we show that the sufficient condition for existence can indeed be interpreted in terms of the interior becoming large while the boundary stays finite, as claimed in the beginning of this section. In \ref{sec:large j limit} we show that the wave front set analysis we present here completely determines the large spin limit. Finally in section \ref{sec:33move} we discuss our results concretely in terms of the 2-complex that occurs in the 3-3 Pachner move. This satisfies our necessary condition for existence and provides an explicit example where most Regge manifolds are suppressed in current spin foam models.

We conclude the paper with discussing the ramifications of our results in section \ref{sec:Discussion}, discussing in particular how the Regge equations of motion can occur in our analysis, how much of our results can survive regularisation and, potentially, renormalisation, and finally, what the presence of the accidental curvature constraints implies for the validity and interpretation of the models.

The appendices contain the derivation of the EPRL and FK wave front sets in appendix \ref{sec:WFS-Lemmas}, the discussion of how the boundary geometry should be interpreted for different boundary Hilbert spaces in appendix \ref{sec:ProjectedSpinNetworkBoundary}, and a detailed discussion of the classical equations of discrete BF theory that also arise in the wave front set analysis in appendix 
\ref{sec:BF}. Finally appendix \ref{sec:conv} contains a brief summary of the conventions  for group elements, Lie algebra elements and coherent states used throughout this paper.

%% file: AD-JHEP-Wojciech.tex

\section{Wave front set calculus}\label{sec:WaveFrontSets}

In order to extract the geometric content of the partition function $\ZZ(\CC)$ we introduce a new tool into the study of spin foam models, the wave front set of a distribution \cite{lars,grigis1994microlocal}. The wave front set is a subspace of the cotangent bundle over the space on which the distribution is defined. Interpreting the distribution as a (generalized) wave function, it can be understood intuitively as the subspace of phase space on which the distribution is peaked in the limit of larger convectors.

We will now give the precise definition. Let $M$ be a smooth compact manifold and ${\calD}(M)$ the distributions over $M$. We denote $\{0\}$ the zero section of the cotangent bundle $T^*M$.
The wave front set $\WF(A) \subset T^*M$ of $A \in \calD(M)$ is defined as the complement of the set
of
elements $\{(x,p)\in T^*M\setminus\{0\}\}$ such that there exists a local coordinate
patch $U\times V$ containing $(x,p)$ with
\begin{equation}
 \forall \phi\in C^\infty_0(U),\ n\in {\mathbb Z}_+\colon\quad
\sup_{\tilde{p}\in V,\ \lambda>0}
\lambda^n \left|\int_U
e^{i\lambda \tilde p\tilde{x}} \phi(\tilde{x})A(\tilde{x})\dd\tilde{x}\right|<\infty
\end{equation}

In other words, the complement of the wave front set are those phase space points at which the limit $\lambda \rightarrow \infty$ falls of faster than any power. The wave front set finds those phase space points that are not suppressed in the limit of large momenta.

We always have $\{0\}\subset \WF(A)$, this is a minor deviation from the conventions in
\cite{lars,grigis1994microlocal}, which always exclude $\{0\}$, that simplifies our bookkeeping.
$\WF(A)$ is a geometric cone in $T^*M$. The singular support of a distribution $A$, 
\begin{equation}
 {\rm sing}\ A=\pi\left(\WF(A)\setminus\{0\}\right)\ ,
\end{equation}
is the projection of nonzero part of $\WF(A)$ into $M$.
Note that the distribution doesn't need to diverge at the singular support. It must, however,
fail to be smooth. The singular support consists exactly of the points where the distribution is not a smooth function.

For a group manifold $G$ (or a product of such) we have a nice coordinate system on $T^*G$. Every
point can be describe by left invariant covector and the group element. We denote such a pair by
$(g,p)\in (G,\mathfrak{g}^*)$.

\subsection{Compositions and multiplications}\label{sec:CompAndMult}

We want to derive the wave front set of the spin foam amplitude \eqref{eq:Hol-SFM-partfunc}, understood as a generalised state in the universal boundary Hilbert space. The partition function is formed by multiplying distributions on group manifolds together, and then integrating in the interior. This is of course not well defined a priori. The spin foam integrand \begin{equation}
\Zi = \prod_{e \in f} E(g_{ef}) \prod_f \omega(\tilde{g}_f)\label{eq:spinfoamintegrand}\end{equation} is the product of distributions on the same base manifold. In order to study when this can actually be defined unambiguously, we will cast their multiplication as restricting their exterior product to the diagonal subset of the direct product of their base manifolds.

That is, given a set of distributions $A^1(m),\dots,A^n(m) \in D(M)$, we consider their exterior product $\prod_{i=1\dots n} A^i(m^i) \in D(M^n)$ which is clearly well defined and then study when this product distribution can be restricted to the diagonal $m^1 = m^i = m$. If possible this then gives the direct product of the distributions. The reason for this strategy is that we can give a criterium for the possibility of this restriction in terms of the wave front set of the exterior product distribution. After restricting in this way we can integrate the product distribution on the interior, and obtain the sought after distribution.

Crucially the wave front sets behave naturally under all these operations, thus this will allow us to give the wave front set of the integrated product of the $A^i$ in terms of the wave front sets of the $A^i$. Applying this to the spin foam partition function we will obtain the wave front set of the spin foam in terms of the wave front sets of $\omega$ and $E$.

In order to disentangle the combinatorial aspects of working on a 2-complex from the juggling of
base manifolds and wave front sets, we will first give the discussion abstractly in terms of
distributions $A$. First, anticipating the structure of the spin foam integrand, note that some of
the distributions $A^i(m)$ might be constant on part of the space $M$. for example, if $M= M_1
\times M_2$, then it might be that we have $A^i(m) = A^i(m_1,m_2) = \tilde{A}^i(m_1)$. In this case
the wave front set of ${A}^i$ is given naturally in terms of that of $\tilde{A}^i$ by the following
property (this is a special case of property \ref{pro-prod} decsribed below):

\begin{pro}[Extension \cite{lars}]\label{pro-ext}
Let $\tilde{A}\in D(M_1)$ then through the projection $r:M_1\times M_2\rightarrow M_1$ we
can define the distribution $ A=r_*\tilde{A}$ by
\begin{equation}
 \int \dd m_1\dd m_2\ \tilde{A}(m_1,m_2) f(m_1,m_2)=\int
\dd m_1\dd m_2\ A(m_1)f(m_1,m_2)
\end{equation}
then
\begin{equation}
 \WF(\tilde{A})=\WF(A)\times \{0\}
\end{equation}
By an abuse of notation we will often denote both distributions by the same letter if it is obvious on which manifold they act\footnote{Generally $M$ might not be of the form $M_1 \times M_2$ but might be a fiber bundle over $M_1$, with projection onto the base $r: M\rightarrow M_1$, then the wavefront set of $\WF(\tilde{A})$ is the pullback of $\WF(A)$.}.
\end{pro}

This allows us to consider all $A^i$ to act on the same space even if the overlap between their
base manifolds is only partial to begin with. Thus, after extension, we can easily give the wave
front set of the exterior product (this is easily derived from the definition for appropriately
chosen coordinates and subsets):

\begin{pro}[Product \cite{lars}]\label{pro-prod}
The wave front set of the exterior product of distributions has the form
\begin{equation}
 \WF(A^1\times\cdots \times A^n)=\WF(A^1)\times\cdots\times \WF(A^n).
\end{equation}
\end{pro}

Note that $\WF(A^1)\times\cdots\times \WF(A^n) \subset (T^*M)^n$ lives in the direct product of the cotangent bundle of $M$.

We can now turn to the issue of restricting this exterior product to the diagonal. We will introduce the base diagonal subspace ${\rm Diag} (T^*M)^n$ of $(T^*M)^n$, given by the elements of $(T^*M)^n$ of the form $(m,m, \ldots,m;p^1,p^2,\dots,p^n)=(m;p^1,p^2,\dots,p^n)$. Similarly $\mathrm{Diag}(\WF(A^1)\times\cdots\times \WF(A^n))$ are the elements of $\WF(A^1)\times\cdots\times \WF(A^n)$ of this form. Generally the diagonal subspace of a subspace $W' \subset (T^*M)^n$, $\mathrm{Diag} (W')$ is the intersection $\mathrm{Diag} (W') = W' \cap {\rm Diag} (T^*M)^n$. We then introduce the restriction operator ${\mathcal D}$,
\begin{equation}
 {\mathcal D}\colon {\rm Diag\ }\left(T^*M\times\cdots\times T^*M\right)
\rightarrow T^*M
\end{equation}
which is defined as
\begin{equation}
{\mathcal D}(m;p_1,\ldots p_n)=(m,p),\quad p=\sum_i p_i\ .
\end{equation}

With this notation in hand we can now give a criterion for when the restriction to the diagonal is
well defined (Beals theorem see \cite{lars}):

\begin{pro}[Restriction \cite{lars}]\label{pro-res}
The restriction of a distribution $A$ on the product $M\times\cdots\times M$ to the diagonal
is well defined if
\begin{equation}
 {\mathcal D}^{-1}(\{0\})\cap\WF(A)=\{0\}.
\end{equation}
In other words, for all $\ww \in {\rm Diag \ } (\WF(A))$ we have that $D(\ww) = 0$ implies that $\ww = 0$.

The restriction of the distribution $A$ to the diagonal then has the wave front set a subset of
\begin{equation}
 {\mathcal D}({\rm Diag }(\WF(A)))
\end{equation}
\end{pro}

Having the wave front set of the restriction in hand, we can note that the wave front set also
behaves naturally under integration. Assume the base manifold $M$ splits into interior and exterior
$M_i \times M_e$. Let us further introduce $r_{ex}$ as the projection from $T^*(M_i \times M_e)$
onto $T^*M_e$, and $r_{in}$ as the projection onto  $T^*M_i$ We then have the following property of
wave front sets under integration (that naturally follows from the definition):

\begin{pro}[Integration \cite{lars}]\label{pro-int}
The wave front set of the distribution $A\in {\mathcal D}(M_i\times M_e)$ integrated over $M_i$
(compact) satisfies
\begin{equation}
 \WF\left(\int A\right)\subset r_{ex}\left(\WF(A)\cap T^*M_e\times \{0\}\right).
\end{equation}
In other words, a necessary though not sufficient condition for $\ww_e \in T^*M_e$ to be in $\WF\left(\int A\right)$, is that there is a $\ww \in \WF(A)$ such that $r_{ex} \ww = \ww_e$ and $r_{in} \ww = 0$.
Moreover, if there is exactly one such $\ww$, then the condition is also sufficient and $\ww_e = r_{ex} \ww$ is in $\WF(\int A)$\footnote{As before this property extends to the case of fibre bundles where the topology of $M$ is not trivial.}.
\end{pro}

We can combine the above properties. To do so it is convenient to introduce the notion of the interior closed subspace. For $W' \subset (T^*M_e \times T^*M_i)$ the interior closed subspace is defined as
\begin{equation}
W'_{\rm icl} = W' \cap \Diag (T^*M_e \times \DD^{-1}(\{0\}_i)).
\end{equation}

That is, it consists of those $(m_i,m_e;p_i^1,p_e^1,\dots,p_i^n,p_e^n)$ that satisfy $\sum_{k = 1\dots n} p_i^k = 0$.   This explains the name interior closed subspace, the interior momenta have to satisfy closure. We further introduce the operator $\partial = r_{ex} \circ \DD$ from $\Diag(T^*M)^n$ to $T^*M_{e}$ We can now state the following Lemma:

\begin{lem}\label{lem:extension and integration}
Given a base manifold $M=M_e \times M_i$ and a set of distributions $A^1,\dots,A^n$ on $M$, then

\begin{equation}
A = \int_{M_i} \dd m_i \prod_{k=1\dots n} A^k
\end{equation}
Is a well defined distribution on $M_e$ if
\begin{equation}
{\partial}^{-1}(\{0\}_e)\cap{(\WF(A^1)\times\dots\times \WF(A^k))_{\rm icl}}=\{0\}.
\end{equation}
or, in other words, for any $\ww \in (\WF(A^1)\times\dots\times \WF(A^k))_{\rm icl}$, $\partial \ww = 0$ implies $\ww = 0$.
Further we have that
\[\WF(A) \subset \partial \left(\WF(A^1)\times\dots\times \WF(A^k)\right)_{\rm icl}\]
\end{lem}

\subsection{Wave front set of the partition function}\label{sec:GeneralPartFunc}

Having all these properties in hand we can now turn towards studying the spin foams partition function directly. To do so in a systematic way we regard the $\omega(g_{f})$ and $E(g_{ef})$ as distributions on all
variables (they are constant in the direction of most of them). That is, we treat them as
distributions on the space
\begin{equation}
\Spin(4)^{\CC_{ev}}\times\Spin(4)^{\CC_{ef}}
\end{equation}
on which the whole integrand is defined. This is our base space $M$. It splits into interior and exterior part as $M_e= \Spin(4)^{\Gamma_{ev}}$ and $M_i = \Spin(4)^{\CC_{ev}/\Gamma_{ev}}\times\Spin(4)^{\CC_{ef}}$.

We will write our extended distributions as $\omega^{(f')}(g_{ev},g_{ef}) = \omega(g_{f'})$ and {similarly} $E^{(e'f')}(g_{ev},g_{ef}) = E(g_{e'f'})$. This is the same extension as discussed in property \ref{pro-ext}.

The exterior product of all the distribution is thus defined on the space
\begin{equation}
(T^*M)^n = \left(\Spin(4)^{\CC_{ev}}\times\Spin(4)^{\CC_{ef}}\right)^{\CC_{f}\cup \CC_{ef}}
\end{equation}

We use the coordinates
\begin{equation}
 g^{(f')}_{ef},g^{(f')}_{ev}, g^{(e'f')}_{ef}, g^{(e'f')}_{ev}
\end{equation}
for this space, where the $g^{(f')}_{ef},g^{(f')}_{ev}$ are associated to the copy of the base space belonging to the distribution $\omega^{(f')}$ that occurs on the face $f'$ in the spin foam partition function and $g^{(e'f')}_{ef}, g^{(e'f')}_{ev},$ are associated to
the distribution $E^{(e'f')}$. Note that we do not assume any a priori connection between the upper label and the lower one.
The corresponding variables in the cotangent bundle are denoted by
\begin{equation}
 p_{ef}^{(f')},p_{ev}^{(f')},p_{ef}^{(e'f')},p_{ev}^{(e'f')}
\end{equation}

By property \ref{pro-prod}, the exterior product of the distributions in the spin foam amplitude, that is, the product of the distributions $\omega^{(f)}$ and $E^{(ef)}$ regarded as living on distinct bases spaces, has the wave front set
\begin{equation}
\begin{split}
W&=\bigotimes_{f} \WF(\omega^{(f)})\otimes \bigotimes_{ef} \WF(E^{(ef)})\\& \subset \left(
\Spin(4)^{\CC_{ev}}\times\Spin(4)^{\CC_{ef}}\times\spin(4)^{\CC_{ev}}\times\spin(4)^{\CC_{ef}}
\right)^{\CC_{ef} \cup \CC_{ev}}
\end{split}
\end{equation}

From property \ref{pro-ext} we immediately have that the elements of the wavefront set $W$ satisfy

\begin{align}\label{pro-1}
 &p_{ef}^{(f')}=0, \text{ if } f\not=f'
&p_{ev}^{(f)}=0 \text{ if } \{e,v\}\notin f \nonumber\\
 &p_{ef}^{(e'f')}=0, \text{ if } e\not=e' \text{ or } f\not=f'
&p_{ev}^{(e'f')}=0
\end{align}

As for the abstract discussion before, we now restrict this product to the diagonal, $\Diag (W)$ in order to obtain the spin foam integrand. The diagonal subspace of $(T^*M)^n$ can be paramterized by $g^{(e'v')}_{ev} = g^{(e'f')}_{ev} = g_{ev}$ and $g^{(e'v')}_{ef} = g^{(e'f')}_{ef} = g_{ef}$ for all $e'$, $f'$, $v'$. Using \eqref{pro-1}, we see that ${\rm Diag}(W)$ exists in a subspace of the whole tangent bundle that can be parametrized by the coordinates
\begin{equation}\label{eq:variables}
 g_{ef}=g^{(f')}_{ef}=g^{(e'f')}_{ef},\quad
 g_{ev}=g^{(f)}_{ev}=g^{(e'f')}_{ev},\quad
 p_{ef}^{(f)},p_{ev}^{(f)},p_{ef} = p_{ef}^{(ef)}
\end{equation}
The operator ${\mathcal D}$,
\begin{equation}
 {\mathcal D}\colon \Diag (W)
\rightarrow T^*\Spin(4)^{\CC_{ev}}\otimes T^*\Spin(4)^{\CC_{ef}}
\end{equation}
then acts as follows on these variables
\begin{align}\label{eq:DDinVariables}
&p'_{ev}=\sum_{f\ni(ev)} p^{(f)}_{ev}
&p'_{ef}=p^{(f)}_{ef}+p_{ef}\nn\\
&g'_{ev}=g_{ev},\
&g'_{ef}=g_{ef}
\end{align}

The properties of the integrand now follow from property \ref{pro-res}. It is well defined if \begin{equation}\DD^{-1}\{0\} \cap \Diag (W) = \{0\},\end{equation} and in that case its wave front set is given by \begin{equation}\WF(\ZZ_{\mbox{int}}) \subset \DD\Diag(W). \end{equation} In other words, using the variables \eqref{eq:variables} and the explicit action of $\DD$ given in \eqref{eq:DDinVariables} it is given by those $(g_{ev}, g_{ef},p'_{ev},p'_{ef})$, such that there exist $p^{(f)}_{ev},  p^{(f)}_{ef}, p_{ef}$ satisfying

\begin{itemize}
\item $p'_{ev}=\sum_{f} p_{ev}^{(f)}$,
\item $p'_{ef}=p^{(f)}_{ef}+p_{ef}$
\item forall $f$, $(g_{ev},g_{ef},g_{ev'},\ldots, p_{ev}^{(f)}, p^{(f)}_{ef},
p_{ev'}^{(f)},\ldots)\in \WF(\omega(\tilde{g}_f))$
\item forall $(ef)$, $(g_{ef},p_{ef})\in \WF(E(g_{ef})).\quad$
\end{itemize}

To obtain the wave front set of the partition function we then turn to Lemma \ref{lem:extension and integration}. To give the concrete statement consider first the form of $W_{\rm icl}$. In terms of our variables this Is the space satisfying

\begin{itemize}
\item $\sum_{f \supset e} p_{ev}^{(f)}=0$, for $(ev) \notin \Gamma_{ev}$
\item $p^{(f)}_{ef} = - p_{ef}$
\item forall $f$, $(g_{ev},g_{ef},g_{ev'},\ldots, p_{ev}^{(f)}, p^{(f)}_{ef},
p_{ev'}^{(f)},\ldots)\in \WF(\omega(\tilde{g}_f))$
\item forall $(ef)$, $(g_{ef}, p_{ef})\in \WF(E(g_{ef})).\quad$
\end{itemize}

This space is essential for what is to follow and thus we give it a new name, ${\WW}(\CC,\Gamma)$. The operator $\partial_{\Gamma}$ now explicitly acts as

\begin{equation}
(g_{ev},g_{ef}, p_{ev}^{(f)}, p^{(f)}_{ef}, p_{ef}) \in {\WW}(\CC,\Gamma) \rightarrow (g_{ev},p_{ev} = \sum_{f \supset e}  p_{ev}^{(f)}) \in T^*\Spin(4)^{\Gamma_{ev}}.
\end{equation}

Then Lemma \ref{lem:extension and integration} states that $\ZZ$ is well defined if
\begin{equation}\label{eq-exists}
\partial_{\Gamma}^{-1}\{0\} \cap \WW(\CC,\Gamma) = \{0\}
\end{equation}
and the wave front set is restricted by
\begin{equation} \WF(\ZZ(\CC,\Gamma)) \subset \partial \WW(\CC,\Gamma)\end{equation}

\subsection{Analysis on group manifolds}\label{sec:Analysis on Group manifolds}

In order to obtain wave front sets of the face amplitudes $\omega$ we need additional properties.
These two can be easily obtain from naturality of wave front set and from the following
useful property, a typical example in books on microlocal analysis {\cite{lars}:}

\begin{pro}[Delta] \label{pro-delta}Let $N\subset M$ be a smooth submanifold. Let
$\delta_N$ be a delta
function of $N$ (with respect to some smooth measure) then
\begin{equation}
 \WF(\delta_N)=\{(x,p)\colon x\in N,\forall_{p^*\in TN} (p,p^*)=0\}\cup \{0\},
\end{equation}
where $(p,p^*)$ is the natural pairing of the vector $p^*$ and the covector $p$.
\end{pro}

By the fact that the wave front set is a geometric cone we also have

\begin{pro}[Parallel transport]\label{pro-parallel} Let $A\in \calD(G)$ where $G$ is a Lie group. Define $\tilde{A}\in\calD(G\times\cdots \times G)$
\begin{equation}
 \tilde{A}(g_1,\ldots, g_n)=A(g_1\cdots g_n)
\end{equation}
then we can show that
\begin{align}
 \WF(\tilde{A})=\{& (g_1,p_1,\ldots,g_n,p_n)\colon
 (g_1\cdots g_n,p_n)\in \WF(A),\nn\\& \forall_{i < n} \;
p_i= g_{i+1} \act p_{i+1}\}
\end{align}
\end{pro}

and also

\begin{pro}[Inverse]\label{pro-inverse} For $A\in \calD(G)$ as above, let
\begin{equation}
 \tilde{A}(g)=A(g^{-1})
\end{equation}
then
\begin{equation}
 \WF(\tilde{A})=\{(g,p)\colon (g^{-1},-g\act p)\in \WF(A)\}
\end{equation}
\end{pro}

These three properties combine naturally to give us the wave front set of the face distribution
$\delta(gg'g''\dots g''')$. Let us first note that the delta function is invariant under the
inverse. So the wave front set of  $\delta(gg'g''\dots g''')$ is the same as  $\delta(g'''^{-1}\dots
g^{-1})$. This operation is the same as switching the fiducial face orientation in our partition
function. To avoid notational overload we will first discuss this wave front set separately. It is
convenient to think of the group elements as parallel transport between fiducial locations $1, 2, 3,
\dots , n, n+1 = 1$. We can then label the group elements going from location $1$ to location $2$ as
 $g_{21}$. The delta function then is of the form $\delta(g_{1n}g_{nn-1}\dots g_{21})$. We then also
have the inverses as $g_{12} = g_{21}^{-1}$. We can see that taking the inverse corresponds
to switching the order of locations in the delta function and we have $\delta(g_{1n}g_{nn-1}\dots
g_{21}) = \delta(g_{12}g_{23}\dots g_{n1})$. Note, however, that the equation in property
\ref{pro-parallel} is always from right to left. That is, we have $p_{32} = g_{21} \act p_{21}$ but
$p_{12} = g_{23} \act p_{23}$. The property \ref{pro-inverse} relates the two by giving us $p_{12} = -
g_{12} p_{21}$. From this we immediately have that $p_{12} = - p_{32}$.  In all the transport
equations, these two lie algebra elements occur only at location $2$, however, they differ in the
orientation in which they appear. 
Let us introduce $p^{2}_{13} = -
p^{2}_{31} = p_{12}$. That is, the upper index gives the location, the lower two indices the
order to which the bivector is associated ($1$ before $3$, or $3$ before $1$ respectively).
 The parallel transport equations now read more naturally as  $p^{2}_{13} = g_{21} \act
p^{1}_{02}$ and $p^{2}_{31} = g_{23} \act p^{3}_{42}$, along with
$p^1_{02} = - p^1_{20}$.

The complete statement then is that $\WF(\delta \prod_i g_{i+1\,i})$ is the set of $g_{i+1\,i}$, and $p^{i}_{i-1,i+1}$ satisfying $\prod_i g_{i+1\,i} = \id$ and $p^{i+1}_{i,i+2} = g_{i+1\,i} \act p^i_{i-1,i+1}$ with the $p^{i}_{i-1,i+1}$ associated to the tangent space of $g_{i+1\,i}$. This notation now is covariant under changing variables to $g_{i\,i+1} = g_{i+1\,i}^{-1}$, with the $p^i_{i+1, i-1}$ now associated to the appropriate tangent space variables.

\subsection{Face distribution}\label{sec:FaceDist}

We can now apply the above to the notationally more complicated case of the spin foam partition function. We assume that $\omega$ has the same
typical form used in all spin foam models and then by \ref{sec:face-wf} it has the same wave front set as the delta function, $\WF(\omega)=\{0,\mathfrak{g}^*\}\cup\{0\}$, and our discussion from above applies directly.

We have the group elements $g_{ev}$, $g_{ve}$ and $g_{ef}$. Whereas the role of the first two as parallel transport between middle of edge and vertex is clear, also notationally, $g_{ef}$ transports from the side of the edge $e$ preceding in the fiducial order $f$ to the the side of the edge succeeding. In order to stay notationally consistent with the previous papers on holonomy spin foams we will keep the notation $g_{ef}$ though.

Consider a face $f$ containing the sequence $(\dots e,v,e',v'\dots)$. The associated product in the face amplitude is $g_{ef}g_{ev}g_{ve'}g_{e'f}g_{e'v'}$. The Lie algebra elements $p^{(f)}_{ef}, p^{(f)}_{ev}, p^{(f)}_{ve'}, p^{(f)}_{e'f},p^{(f)}_{e'v'}$ are on the right of their respective group elements, that is, $p^{(f)}_{ev}$ lives at the vertex $v$,  it is transported by $g_{ev}$ to $p^{(f)}_{ef}$, which lives at the middle of the edge, but before $g_{ef}$, and so on. The transport equations are then

\begin{align}
p^{(f)}_{ef} =  g_{ev} \act p^{(f)}_{ev} ,&\quad p^{(f)}_{ev} =  g_{ve'} \act p^{(f)}_{ve'},\nn\\
p^{(f)}_{ve'} = g_{e'f} \act p^{(f)}_{e'f} ,&\quad p^{(f)}_{e'f} =  {g_{e'v'}} \act p^{(f)}_{e'v'}.
\end{align}

This suggests, as in the abstract discussion above, to change notation to respect the "location" of the Lie algebra elements. We accomplish this by writing $p^{(f)}_{ev} = p^v_{ee'}$, where the $ee'$ plays the role of both indicating the particular face, and the orientation in it. We want to keep $p^{(f)}_{ef}$, and always regard it as living before the group element $g_{ef}$ in the fiducial orientation of $f$. We denote $\tilde{p}^{(f)}_{ef}$ the Lie algebra element to the left of $g_{ef}$. Rewritten this way our transport equations become

\begin{align}
p^{(f)}_{ef} =  g_{ev} \act p^{v}_{ee'} ,&\quad p^{v}_{ee'} =  g_{ve'} \act \tilde{p}^{({f})}_{e'{f}},\nn\\
\tilde{p}^{({f})}_{e'{f}} = g_{e'f} \act p^{(f)}_{e'f} ,&\quad p^{(f)}_{e'f} = {g_{e'v'}} \act p^{v'}_{e'e''}.
\end{align}

\begin{figure}[htp]
 \centering
 \includegraphics[scale=.90]{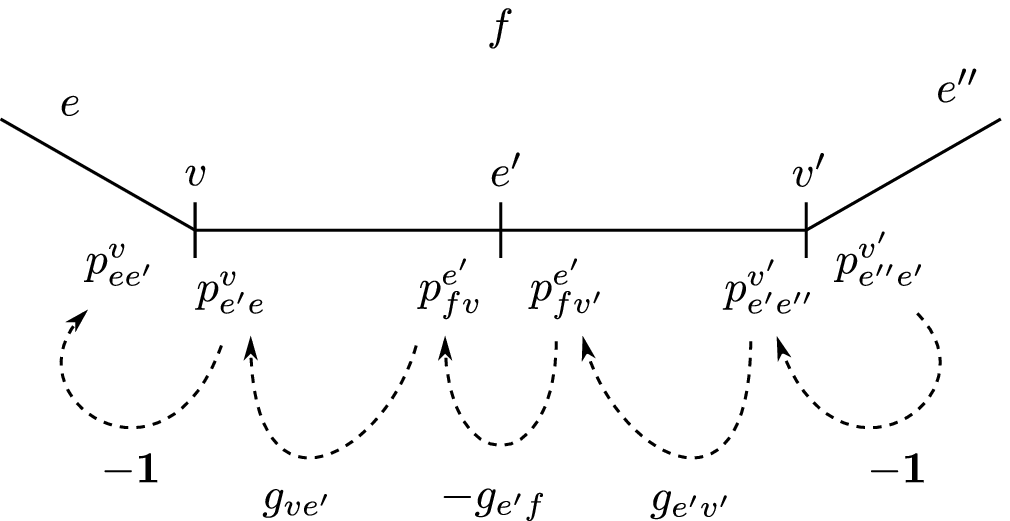}
 \caption{The covectors at edge $e'$ in a face $f=(\dots ,e, v,e',v',e'',\dots)$.}
 \label{fig-data0}
\end{figure}

We can make these equations (almost) independent of the fiducial face orientation by making the position of the $p^{(f)}_{ef}$ and $\tilde{p}^{(f)}_{ef}$, notationally explicit, and introducing the $p^v_{e'e} = -p^v_{ee'}$ associated to the opposite orientation, in analogy to the left and right arrows above. To do so we introduce $p^{e'}_{fv'} = p^{(f)}_{e'f}$ and $p^{e'}_{fv} = - \tilde{p}^{(f)}_{e'f}$. In this way the parallel transport equations simplify to the following three equations:

\begin{align}
p^{v}_{e'e} =  g_{ve} \act p^{e}_{fv} ,&\quad p^{e}_{fv} = - g_{ef} \act p^{e}_{fv'} \nn\\
p^v_{e'e} = -p^v_{ee'}.&
\end{align}

The whole set of covectors and group elements, along with the equations is given in figure \ref{fig-data0}. What is subtle in this picture is the identification with the cotangent variables so we summarize this again explicitly here. For a sequence $(v,e',v',e'')$ in a face $f$, as in figure \ref{fig-data0}, we have

\begin{align}
p^{(f)}_{e'v'} = p^{v'}_{e'e''},&\quad  p^{(f)}_{ve'} = p^{e'}_{fv}\nn\\
\text{and} & \quad p^{(f)}_{e'f} = p^{e'}_{fv'}\ .
\end{align}

Keeping these identifications in mind, the set of conditions for ${\mathcal W}(\CC,\Gamma)$ can now be written as follows
\begin{equation}
 g_{ef},g_{ev}, p^v_{ee'}, p^e_{fv}
\end{equation}
satisfy
\begin{itemize}
\item (transport equations)
\begin{align}
p^{v}_{e'e} =  g_{ve} \act p^{e}_{fv} ,&\quad p^{e}_{fv} = - g_{ef} \act p^{e}_{fv'} \nn\\
p^v_{e'e} = -p^v_{ee'}.&
\end{align}
\item (simplicity constraints) $(g_{ef},-p^{e}_{fv})\in \WF(E)$ for $(ev) \subset f$,
\item (holonomy) $\tilde{g}_f=\id$ unless any $p^{(f)}_{\cdot}=0$. In the case $p^{(f)}_{\cdot}=0$,
$g_{ef}$ and $g_{ev}$ are not restricted. Both $\tilde{g}_f$ and $g_f$ are free, as well.
\item (closure) for all $e\notin\Gamma_e$, and $v \in e$ we have $\sum_{f\ni e} p^{e}_{fv}=0$.
\end{itemize}
If $p^{(f)}_{\cdot}=0$, the conditions from $\WF(E)$ are also vacuous.

\section{BF theory and the Barrett-Crane model}\label{sec:BFandBC}

The remaining distribution left to analyse is the E-function. In the case of BF theory and the Barrett-Crane model the E-function is just a delta function again, and we need no further analysis to describe their wave front sets.

\subsection{BF theory}
As a consistency check we apply our method to BF theory. The $E$ function in this case is just
\begin{equation}
 E(g)=\delta(g)
\end{equation}
thus its wave front set is $\{\id,\mathfrak{g}^*\}\cup\{0\}$.

\subsubsection{Amplitude}

The simplicity condition in ${\mathcal W}_{\text{BF}}(\CC,\Gamma)$ is thus just
\begin{equation}
  g_{ef}=1,\quad p^{e}_{fv}=-p^{e}_{fv'}
\end{equation}
Thus the equations simplify to $g_f=1$, closure, and the transport equations. We can eliminate
$p^{e}_{fv}$ and obtain the following criterion that determines points of ${\mathcal
W}(\CC,\Gamma)$
\begin{itemize}
 \item $g_f=\id$, $g_{ef}=\id$
\item $p^{v}_{ee'}$ are pararell transported by $g_{ve} g_{ev'}$,
\item $\forall_{e\notin\Gamma}$ and $v \in e$, $\sum_f p^{e}_{fv}=0$
\end{itemize}
Following the same ideas as \cite{Barrett:2008wh} for the special case of the three dimensional $\su(2)$ BF theory (the Ponzano-Regge model), we can describe the image of the map $\partial$ as follows:
\begin{itemize}
 \item In the $2$-complex $\CC$ we have a well defined notion of contractible loop. For every such  {contractible loop in the boundary} $l=\{vev'e'\cdots\}$ we must have, \begin{align}
 g_l=g_{ve}g_{ev'}g_{v'e'}\cdots=1\ .
\end{align}
Such a configuration on the boundary can always be extended into the interior.
\item For $\Gamma$ we can introduce the same boundary operator in local homology as \cite{Barrett:2008wh}:
\begin{equation}
 d_1\colon C(\Gamma_e)\rightarrow C(\Gamma_v),\quad d_1(\mathcal{E})(v)=\sum_{e\ni v} \pm g_{ve}\rhd \mathcal{E}(e)
\end{equation}
where $C(\Gamma_e)$ and $C(\Gamma_v)$ are linear vector spaces generated by sets of $\su(2)$ labelled by $\Gamma_e$ and $\Gamma_v$, and $\mathcal{E} \in C(\Gamma_e)$.
The sign depends on the position of the vertex in relation to the orientation of the edge. We have natural map of local homology groups
\begin{equation}
 {\rm [incl]}: H_1(\Gamma,\su(2))=\ker d_1\rightarrow H_1(\CC,\su(2))
\end{equation}
The condition on bivectors is that the boundary $\pm p_{ev}$ (sign depending again on the orientation of the edge) should belong to ${\rm [incl]}^{-1}(0) \subset \ker d_1 \subset C(\Gamma_e)$ (see section \ref{sec:BF}).
\end{itemize}
These are indeed the equations of discretised BF theory. Thus we really recover the classical BF equations of motion with this method, validating that it is, at least in this case, a semiclassical limit. This can be seen as a generalisation of the asymptotics of the partition function of the Ponzano-Regge model derived in \cite{Dowdall:2009eg}.

Of course this wave front set is valid only if the conditions of Lemma \ref{lem:extension and integration} are satisfied
\footnote{i.e. $H_2(\CC,\su(2))=\{0\}$ for all flat parallel transports agreeing with boundary holonomy, see section \ref{sec:BF}.}. This excludes for example all foams with internal bubbles. This is not surprising as these foams can be infinite without regularization. These conditions are exactly satisfied by the no-tardis triangulations of \cite{Barrett:2008wh}.

\subsection{The BC model}\label{sec:BC}

The E-function of the Barrett-Crane model is also a delta function. To describe its support we need to go a bit deeper into the structure of the gauge group $\Spin(4)$ and its Lie algebra $\spin(4)$. Recall that $\spin(4)=\su(2)\oplus\su(2)$. We thus can write
\begin{equation}
 p=(\wec{p}^+,\wec{p}^-)
\end{equation}
where $\wec{p}^\pm$ can be identify with either combination of Pauli matrices or with three dimensional
vectors. In this form these Lie algebra elements generate group elements $g = (g^+,g^-)$, and the E-function can be written as 
\begin{equation}
 E(g)=\delta(g^+(g^-)^{-1}).
\end{equation}
The wave front set is described by
\begin{equation}
 g^+=g^-,\quad p=(\wec{n},-\wec{n}),\ n\in \su(2)
\end{equation}
So $g\act p\not= p$ in general. To describe this wave front set more geometrically let us introduce the 4-vector $N^0 = (1,0,0,0)$ and the Hodge dual $\hodge (\wec{p}^+,\wec{p}^-) = (\wec{p}^+,-\wec{p}^-)$ on the Lie algebra. Interpreting the Lie algebra elements $p$ as bivectors on $\R^4$ we can contract them with vectors. The condition $\exists n$ s.t. $p=(\wec{n},-\wec{n}),\ \wec{n}\in \su(2)$ can be written geometrically as $N^0 \cdot \hodge p = 0$. For more details on this bivectorology see e.g. \cite{Barrett:1997gw,Freidel:2007py,Barrett:2009gg,Hellmann:2010nf}.

\subsubsection{Amplitude}

The condition for ${\WW}(\CC,\Gamma)$ coming from $(g_{ef},-p^{e}_{fv})\in \WF(E)$ is
\begin{equation}
 g_{ef}^+=g_{ef}^-,\ N^0 \hodge p^e_{fv} = 0\ .
\end{equation}
As opposed to the case of BF theory and also the EPRL and FK model that we will analyse next, the Barrett-Crane model has $p^{e}_{fv}\not=-p^{e}_{fv'}$ because $g_{ef}$ does not in general stabilize
$p^e_{fv}$.
For every internal edge
\begin{equation}
 \sum_f p^e_{fv}=0.
\end{equation}
Moreover the $p^e_{fv}$ and $p^v_{ee'}$ satisfy the transport conditions and are simple bivectors. Anticipating the full geometric discussion that will follow in section \ref{sec:Geometry}, we can see that the $p^v_{ee'}$ satisfy the necessary conditions for reconstructing a 4-simplex geometry at the vertex $v$. The $p^e_{fv}$ then encode the boundary geometry. The fact that $p^e_{fv} \neq - p^e_{fv'}$ and in particular that there are distinct rotations $g_{ef}$ acting on the different faces, means that the boundary geometry at neighbouring 4-simplices does not necessarily agree. This is known as the ultra-locality problem of the Barrett-Crane model.

\section{The EPRL and FK model}\label{sec:EPRL-FK E func}

While for BF theory and the Barrett-Crane model the only distributions involved were $\delta$ functions, the EPRL and FK model have much more complicated distributions. In order to analyse them we will need to dvelve deeper into the structure of wave front sets.

\subsection{Actions and differential equations}\label{sec:Prop Actions and diff equations}

In order to derive the wave front set of EPRL/FK type simplicity functions we need more properties of the
wave front sets. One of the most important facts about microlocal analysis, see \cite{lars}, is the
following:

\begin{pro}[Invariance \cite{lars}]
Let $C$ be a pseudodifferential operator on $M$ and $c$ its principal
symbol (a homogeneous function on $T^*M$). If for $A\in\calD(M)$
\begin{equation}
 CA=0\text{ or } CA \text{ is smooth}
\end{equation}
then
\begin{equation}
 \WF(A)\subset \{c=0\}\cup \{0\}
\end{equation}
where $\{c=0\}$ is the subset of $T^* M\setminus \{0\}$ on which principal
symbol vanishes.
Wherever on $\{c=0\}$ the exterior derivative of $c$, $dc$, does not vanish ( $dc\not=0$), the set $\{c=0\}$ it is also invariant under the
hamiltonian flow generated by $c$.
\end{pro}

Thus we can deduce

\begin{pro}[Group invariance]
Let $G$ acts smoothly on $M$, then we have also a natural action
(symplectic) on $T^*M$. If $A\in\calD(M)$ is invariant under $G$ then
$\WF(A)$ is also invariant. If $G$ is the Lie group and the action is generated
by the vector fields $L\in {\mathfrak g}$ then
\begin{equation}
 \WF(A)\subset \{(x,p)\in T^*M\colon \forall_{L\in{\mathfrak g}}\ (p,L)=0\}
\end{equation}
where $(p,L)$ is the natural pairing of vectors and covectors.
\end{pro}

\subsection{The EPRL function}\label{sec:EPRL function}

We now turn towards the EPRL E-function. In order to apply the above properties of wave front sets we need to characterize the E-function in terms of its symmetries and through differential operators. As explained in more details in section \ref{sec-E-der}, the wave front set of the EPRL E-function is almost uniquely determined by the following properties:
\begin{enumerate}
\item $E$ is invariant under adjoint $SU(2)$ action
\item $E$ is the solution of $CE=0$
where
\begin{equation}
  C=|1-\gamma|\left(\sqrt{C_++\frac{1}{4}}-\frac{1}{2}\right)-(1+\gamma)\left(\sqrt{C_-+\frac{1}{4}}-\frac{1}{2}\right)
\end{equation}
where $C_\pm$ are the Casimirs of $\SU(2)_\pm$.
{This is the differential equation imposing the condition $|1-\gamma|j^+=(1+\gamma)j^-$ because
\begin{equation}
 \sqrt{C_\pm+\frac{1}{4}}-\frac{1}{2},
\end{equation}
return the value $j^\pm$ on irreducible representation.}
\item $E$ is also the solution of $\tilde{C}E=0$ where
\begin{equation}
 \tilde{C}=2\sqrt{C_++\frac{1}{4}}-(1+\gamma)\sqrt{L^2+\frac{1}{4}}+\frac{1}{2}(\gamma-1)
\end{equation}
where $L^2=(L_i^++L^-_i)^2$.
{This is the equation imposing the condition $2j^+=(1+\gamma)k$.}
\end{enumerate}

In order to give the EPRL wave front set geometrically we introduce the map $T_{\gamma} = 1+\gamma \hodge$ on the Lie algebra. This has the property that $T_\gamma T_{-\gamma} = 1-\gamma^2$. We are always away from $\gamma = 1$.

Theorem \ref{theo-eprl-unique}, proven in section \ref{sec-E-der}, then states that

\begin{theo}\label{theo-WF1}
The wave front set of a distribution satisfying the conditions above is the sum of sets of the type $W_\alpha$
($\alpha\in[0,\pi]$) defined by
\begin{align}
 W_\alpha=\left\{\left(e^{\alpha(\wec{n},\wec{0})+\xi\hodge T_{-\gamma}(\wec{n},\wec{n})}, T_\gamma(\wec{n},\wec{n})\right)\colon
\xi\in {\mathbb R},\wec{n}\in\su(2)\right\}
\end{align}
and the zero section $\{0\}$.
\end{theo}

In the case $\gamma<1$ we know the the E function can be written as \footnote{Our convention here differs from that of \cite{Bahr2013,Dittrich2013} in that we absorb the dimension of the representation into the edge weight $d_e$, this will simplify the notation for the derivation of the wave front set.}\begin{equation}\label{formula-EPRL}
 E^\gamma_{\text{EPRL}}(g)=\sum_{k\in {\mathbb N}k_0} d_e(k)
\int_{S^2} \dd\wec{n}\ \langle n|g^+|n\rangle^{(1+\gamma)k}\langle
n|g^-|n\rangle^{(1-\gamma)k}
\end{equation}
where $k_0$ is the minimal $k$ such that both $j^\pm=\frac{1}{2}(1\pm\gamma)k$ are half-integer or integer.

Because $d_e(k)$ has at most polynomially growth, this is regular outside the set
where there exists an $\wec{n}$, such that
\begin{equation}
 \langle n|g^+|n\rangle^{(1+\gamma)k}\langle
n|g^-|n\rangle^{(1-\gamma)k}=1 \ .
\end{equation}
If this does not exist the summands fall off exponentially and the sum is a smooth function.
Since $2j^\pm_0$ are integer and $g^\pm$ stabilise $n$, the condition
\begin{equation}
 2j^+_0\ln\langle n|g^+|n\rangle+2j^-_0\ln\langle -n|g^-|-n\rangle=0\text{ mod } 2\pi
\end{equation}
can be rewritten as 
\begin{equation}\label{eq:theta-eprl}
 2j^+_0\frac{\theta^+}{2}+2j^-_0\frac{\theta^-}{2}=2\pi m,\quad m\in {\mathbb N}
\end{equation}
where $g^\pm=e^{\theta^\pm \wec{n}}$.  Notice that $2j^\pm_0$ are integer, $2j^+_0+2j^-_0=2k_0$ and $k_0$ is minimal thus $2j^\pm_0$ have no common divisors. We can always write $m=a_+\ 2j^++a_-\ 2j^-$ for some $a_\pm\in{\mathbb Z}$, then
\begin{equation}
 2j^+_0(\theta^+-4\pi a_+)+2j^-_0(\theta^--4\pi a_-)=0
\end{equation}
Thus, because angles are only defined modulo $4\pi$ 
it is enough to solve the equation \eqref{eq:theta-eprl} with $m=0$:
\begin{equation}
 ({\theta^+},{\theta^-})=\xi((1-\gamma),-(1+\gamma))\ .
\end{equation}
The singular support of the distribution must belong to
\begin{equation}S=\left\{(g^+,g^-)\colon \exists \wec{n}, {\xi}, g^\pm=e^{\pm
{\xi}(1\mp\gamma)\wec{n}}\right\}\ ,
\end{equation}
We know that 
\begin{equation}
 \pi(\WF(E^\gamma_{\text{EPRL}})\setminus\{0\})={\rm sing}\ E^\gamma_{\text{EPRL}}\subset S
\end{equation}
Combining this with theorem \ref{theo-WF1} and the fact that only for $\alpha=0$, $\pi(W_\alpha)\subset S$ we obtain
\begin{equation}
 \WF(E^\gamma_{\text{EPRL}})=W_0\cup\{0\}, \text{ or } \WF(E^\gamma_{\text{EPRL}})=\{0\}
\end{equation}
Given a normalised vector in $v_k\in V_{j^+,j^-}$ such that $L^2v_k=k(k+1)v_k$ we have a function
\begin{equation}
 \psi_k(g)= \langle v_k|D_{j^+,j^-}(g) |v_k\rangle
\end{equation}
From \cite{Bahr2013} we know that
\begin{equation}
\la \overline{\psi}_k|E^\gamma_{\text{EPRL}} \ra = {\frac{d_e(k)}{(2j^++1)^2(2j^-+1)^2}}\ .
\end{equation} 
If $d_e(k)^{-1}$ also grows at most polynomially, then by the condition given in section \ref{sec:smooth}, $E^\gamma_{\text{EPRL}}$ is not smooth.
Thus the only possibility is
\begin{equation}
 \WF(E^\gamma_{\text{EPRL}})=\left\{\left(e^{ \xi \hodge T_{-\gamma}\hat{p}}, T_{\gamma}p\right)\colon p=(\wec{n},\wec{n}),\ \wec{n}\in su(2)\right\}\cup
\{0\}.
\end{equation}
We introduce the notation for normalized bivectors
\begin{equation}
 \hat{p}=\frac{p}{|p|}
\end{equation}
Using the same more geometric notation as before,
\begin{equation}
 \WF(E^\gamma_{\text{EPRL}})=\left\{\left(e^{ \xi \hodge 
R_{\gamma}\hat{p}}, p\right)\colon \  N^0 \cdot T_{-\gamma}p=0 \right\}\cup
\{0\},
\end{equation}
{where in the last line we introduced the shortcut
\begin{equation}\label{eq:Rgamma}
 R_\gamma p=\frac{1}{\sqrt{1+\gamma^2}} T_{-\gamma}T_{\gamma}^{-1} p\ .
\end{equation}
Note that for normalized bivectors satisfying simplicity $N^0\cdot p=0$ we have 
\begin{equation}
 |T_\gamma p|=\sqrt{1+\gamma^2}|p|
\end{equation}
as can be seen by considering the action on the left and right part of the bivector separately.

A similar, slightly more complicated formula as \ref{formula-EPRL} is also valid for $\gamma>1$, and the analogous argument holds in that case.}

As opposed to the case of the Barrett-Crane model the wave front set contains only a subspace of all covectors annihilating the tangent space of the critical manifold. This
subset consists of those annihilating covectors that are stabilized by the group element {
\begin{equation}
 (g,p)\in \WF(E^\gamma_{\text{EPRL}})\setminus\{0\}\Leftrightarrow 
\left\{ \begin{array}{c}
         g\in {\rm sing}\ E^\gamma_{\text{EPRL}}\\
         g\act p=p\\
         p\perp T{\rm sing}\ E^\gamma_{\text{EPRL}}
        \end{array}
\right\}
\end{equation}
where} ${\rm sing}\ E^\gamma_{\text{EPRL}}$ is singular support of the distribution $E^\gamma_{\text{EPRL}}$
{ and $T{\rm sing}\ E^\gamma_{\text{EPRL}}$ is the tangent space to ${\rm sing}\ E^\gamma_{\text{EPRL}}$.}

\subsection{The FK model $\gamma>1$}\label{sec:The FK model gamma larger 1}

For $\gamma < 1$ the EPRL and the FK model coincide exactly. For $\gamma>1$ the FK E-function is given by
\begin{equation}
E_{FK}^\gamma(g)=\sum_{k\in{\mathbb N}k_0} d_e(k)\int \dd g'\ \langle
n|{g'}^{-1}g^+g'|n\rangle^{(1+\gamma)k}\langle -n|{g'}^{-1}g^-
g'|-n\rangle^{(\gamma-1)k}. 
\end{equation}

To understand the FK model for $\gamma > 1$ we need to first study the distribution \begin{equation}
E_n(g)=\sum_k d_e(k) \langle
n|g^+|n\rangle^{(1+\gamma)k}\langle -n|g^-|-n\rangle^{(\gamma-1)k}.
\end{equation}
Where $d_e(k)$ is an edge weight of at most polynomial growth. We will again write $j^\pm = \frac12|1\pm\gamma|k$. To understand this distribution using our properties from above we can write it using
coherent state kernel distribution \ref{sec:coherent_state_dist_kernel}
\begin{equation}
K_n(x,g)=\sum_{k\in {\mathbb N}k_0} e^{-i{\frac{k}{k_0}}x}\langle
n|g^+|n\rangle^{2j^+}\langle -n|g^-|-n\rangle^{2j^-}
\end{equation}
as 
\begin{equation}
 E_n(g)=\int\dd x\ d_e\left(i\frac{\partial}{\partial x}\right) K_n(x,g)
\end{equation}
We will prove in Lemma \ref{lem-FK} that
\begin{equation}
\begin{split}
 \WF(K_n){ =}\Big\{&
\left({2j_0^+}\ln\langle n|g^+|n\rangle+{2j_0^-}\ln\langle -n|g^-|-n\rangle,
g^+,g^-; -\xi,\wec{p}^+,\wec{p}^-\right)\colon \\
&\xi>0,\ \wec{n}=g^\pm\akt \wec{n},\
\wec{p}^\pm=\pm{ 2j_0^\pm}\xi \wec{n} \Big\}\cup \{0\}
\end{split}
\end{equation}
{ where $j^\pm_0=\frac{|1\pm\gamma|}{2}k_0$. }
Using this wave front set, together with the ellipticity property \ref{sec:ellipticity} and integration formulas  \ref{pro-int}  (with $r_{ex}$ injective) we can show that
\begin{equation}
\begin{split}
 \WF(E_n)= \{&(g,p)\colon { 2j_0^+}\ln\langle n|g^+|n\rangle+{ 2j_0^-}\ln\langle -n|g^-|-n\rangle=0 \mod 2\pi,\\
&\wec{p}^\pm=\pm{ 2j_0^\pm} \wec{n},\ \wec{n}=g^\pm \wec{n}\}
 \cup \{0\}\ .
\end{split}
\end{equation}
Since $2j^\pm_0$ are integer and $g^\pm$ stabilize $n$, the condition
\begin{equation}
 2j^+_0\ln\langle n|g^+|n\rangle+2j^-_0\ln\langle -n|g^-|-n\rangle=0\text{ mod } 2\pi
\end{equation}
can be rewritten as in \eqref{eq:theta-eprl} as
\begin{equation}
 2j^+_0\frac{\theta^+}{2}+2j^-_0\frac{\theta^-}{2}=2\pi m,\quad m\in {\mathbb N}
\end{equation}
where $g^\pm=e^{\theta^\pm \wec{n}}$. As in \eqref{eq:theta-eprl} we have
\begin{equation}
 ({\theta^+},{\theta^-})=\xi((1-\gamma),-(1+\gamma))
\end{equation}
The only solution of the conditions for $g^\pm$  is  $g=e^{\xi\hodge T_{-\gamma}(n,n)}$. The wave front set can be stated as
\begin{equation}
 \WF(E_n){ =} \{(g,T_\gamma p)\colon g=e^{\xi\hodge T_{-\gamma}{ \hat{p}}},\ p=(\wec{n},\wec{n})\}\cup\{0\}
\end{equation}
where as before $\hat{p}$ is the normalized $p$ bivector.

We can now compute in stages
\begin{equation}
 E^\gamma_{\text{FK}}(g)= \int_{S^2} \dd\wec{n}\ A(n,g),\quad 
A(n=[g'],g)= \int \frac{\dd\theta}{2\pi} A(e^{\theta \wec{n}} g'g(g')^{-1}e^{-\theta \wec{n}})
\end{equation}
Using in stages the properties  \ref{pro-prod} and \ref{pro-int} we get:
\begin{equation}\label{eq:FKWFS-conditions}
 \{(g,T_\gamma p)\colon \exists g',\ g'g(g')^{-1}=e^{\xi\hodge g'\act T_{-\gamma}\hat{p}}, g'\act T_{-\gamma}p=(\wec{n},\wec{n}),\
T_{\gamma} p\perp (L_i-g\act L_i)\}\cup\{0\}
\end{equation}
The only point is to prove that the application of the property \ref{pro-int} we obtain always the whole set.
For integration over $n$ the argument is that the projection is unique, whereas for integration over $\theta$ we can use 
\begin{equation}
 A(e^{\theta \wec{n}} g e^{-\theta \wec{n}})=A(g)
\end{equation}
Thus $A(g'g(g')^{-1})$ is the pullback of $A(n,g)$. Using property \ref{pro-ext} we can deduce equality also in this case.

The last condition of \eqref{eq:FKWFS-conditions} is tautologically satisfied, so we have
\begin{equation}
 \WF(E^\gamma_{\text{FK}}) = \{(g,T_\gamma p)\colon \ g=e^{\xi\hodge T_{-\gamma}\hat{p}}, p=(\wec{m},\wec{m}), \wec{m}\in \su(2)\}\cup\{0\}
\end{equation}
The wave front set is thus the same as in the case of the EPRL model. 
\begin{equation}
 \WF(E_{\text{FK}}^\gamma)= \{(g,p)\colon \ g=e^{\xi\hodge R_{\gamma}\hat{p}}, N^0\cdot T_{-\gamma}p = 0\}\cup\{0\}\ .
\end{equation}

\subsection{The $\overline{\mbox{FK}}$ model}\label{sec:FK bar}

The same analysis may be applied to the model Freidel and Krasnov proposed for the state sum without $\gamma$. We call this the $\overline{\mbox{FK}}$ model as it differs from the $\gamma = 0$ version by a complex conjugation in the $g^-$ sector. Explicitly its $E$ function is given by
\begin{equation}
 E_{\overline{\text{FK}}}(g)=\sum_k d_e(k)\int \dd g\ \langle
n|g^{-1}g^+g|n\rangle^{k}\langle -n|g^{-1}g^-g|-n\rangle^{k}. 
\end{equation}
The result is
\begin{equation}
 \WF(E_{\overline{\text{FK}}})= \{(g,p)\colon \ g=e^{\xi\hodge \hat{p}}, p=(\wec{m},-\wec{m}), \wec{m}\in \su(2)\}\cup\{0\}
\end{equation}

\subsection{Wave front set of the amplitude}\label{sec:EPRL-FK partition function wave front set}

We can now give the conditions on${\mathcal W}(\CC,\Gamma)$ for the various choices of E function given above. Note that as EPRL and FK have the same wave front sets, depending only on $\gamma$, we have the condition that $g_{ef}$ stabilizes $p_{ef}^{(f)}$ so we have
\begin{equation}
 p^{e}_{fv}=-p^{e}_{fv'}
\end{equation}
as in BF theory and in contrast to the BC model. Thus the transport equations split into two orientation equations and one parallel transport equation. We can thus summarize the conditions for configurations in ${\mathcal W}(\CC,\Gamma)$ as follows:
\begin{itemize}
 \item orientation and parallel transport equations:
\begin{align}
p^v_{e'e} = -p^v_{ee'} ,&\quad p^{e}_{fv} = - p^{e}_{fv'} \nn\\
p^{v}_{e'e} =  g_{ve} \act p^{e}_{fv}\ .&
\end{align}
\item simplicity: $N^0 \cdot T_{-\gamma} p^{e}_{fv} = 0$
\item closure: for all $e \notin \Gamma$, and $v \in e$ we have $\sum_{f\ni\ e} p^{e}_{fv}=0$.
\item holonomy: $\tilde{g}_f=1$ unless $p^{e}_{fv}=0$, otherwise no restriction,
\item $g_{ef}=e^{\xi_{ve}\hodge {R_\gamma\hat{p}^{e}_{fv}}}$ unless $p^{e}_{fv}=0$, otherwise
no restriction,
\end{itemize}
{With the notation of \eqref{eq:Rgamma}.}

To finally state or main result on the behaviour of the amplitude we then need to give the action of $\partial$ on these variables. In the 2-complexes we consider every boundary edge is 1-valent, that is, it borders exactly one face. Given a face $f$ containing a boundary edge $e \in \Gamma_e$, $f = ( \dots v,e,v',e' \dots)$ we have that $\partial_e$, that is, $\partial$ restricted to this boundary edge acts as follows on an element $\ww$ in $\WW(\CC,\Gamma)$:
\begin{align}
\partial_e:& \WW(\CC,\Gamma) \rightarrow T^* \Spin(4)_{ve} \times T^* \Spin(4)_{ev'}\nn\\
& \partial_e(\ww) = (g_{ve},-p^e_{fv},g_{ev'},p^{v'}_{ee'})
\end{align}
The operator $\partial_{\Gamma}$ is then just the product of the operators $\partial_e$, \begin{equation} \partial_{\Gamma}(\ww) = \times_{e \in \Gamma_e} \partial_e(\ww)\ .\end{equation}
We see that under $\partial_{\Gamma}$, $p_{ev}$ is twisted simple.

Changing variables from $g_{ve}$ to $g_{ev}$ we have
\begin{align}
\partial_e:& \WW(\CC,\Gamma) \rightarrow T^* \Spin(4)_{ve} \times T^* \Spin(4)_{ev'}\nn\\
& \partial_e(\ww) = (g_{ev},p^v_{ee''},g_{ev'},p^{v'}_{ee'})
\end{align}
{It can be written simply as $\partial_{e}=\partial_{ev}\times\partial_{ev'}$ where}
\begin{align}\label{eq:gev-bdry}
\partial_{ev}:& \WW(\CC,\Gamma) \rightarrow T^* \Spin(4)_{ev}\nn\\
& \partial_{ev}(\ww) = (g_{ev},p^v_{ee'})
\end{align}

Having the full set of equations for $\WW(\CC,\Gamma)$ we can now state our first main result of the paper
\begin{prop}\label{prop:first-main-result}
With $\partial_\Gamma$ and $\WW(\CC,\Gamma)$ given as above, we have that the partition function $\ZZ(\CC,\Gamma)$ exists as a distribution if
\begin{equation}\label{eq-exists2}
\partial_{\Gamma}^{-1}\{0\} \cap \WW(\CC,\Gamma) = \{0\}.
\end{equation}
In that case, its wave front set is restricted by 
\begin{equation} \WF(\ZZ(\CC,\Gamma)) \subset \partial_\Gamma \WW(\CC,\Gamma)\end{equation}
\end{prop}

%% file: AD-JHEP-33Move.tex

\subsection{A concrete case: The 3-3 move}\label{sec:33move}
An interesting example to consider is that of the 2-complex that arises in the 3-3 Pachner move, $\CC_{33}$. This is dual to three 4-simplices glued together around a common face. The triangulation has no internal edges, but an internal face. Thus there are no Regge equations of motion, on this triangulation all discrete metrics are Regge metrics. This example has been studied repeatedly in the literature \cite{Mamone:2009pw,Magliaro:2011qm}. 

\begin{figure}[htp]
 \centering
 \includegraphics[scale=.70]{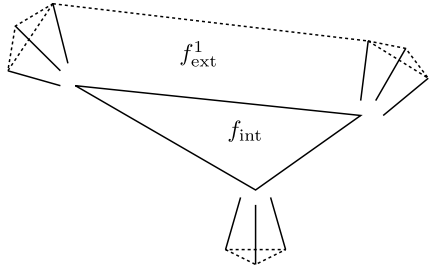}
 \caption{The dual complex of the configuration that occurs in the 3-3 move, with internal edges solid and some of the boundary edges given as dotted lines. The only face not touching the boundary, and thus made of fully internal edges is $f_{\text{int}}$. The faces touching the boundary come in two flavours, belonging to one 4-simplex or belonging to two. The face $f^1_{\text{ext}}$ is an example of a face belonging to two simplices.}
 \label{fig-33move}
\end{figure}

The 2-complex of the 3-3 move consists of a single internal face $f_{\text{int}}$ with three boundary edges. Each of these three boundary edges has three more faces associated to it. These faces all touch the boundary. Each of the vertices has four more edges that also go out to the boundary. Let us consider the conditional above. The map $\partial$ identifies the boundary bivectors with the bivectors of the faces that touch the boundary. Thus the conditional reads that if all the bivectors on faces that touch the boundary are set to zero, then the internal faces are forced to be zero, too.

This follows simply from the closure condition on the edges of the internal face. As all terms but the bivector belonging to the internal face are set to zero closure simply reads $p^e_{f_{\text{int}}v} = 0$. Clearly this argument can be iterated. 
This leads to the notion of a no-tardis triangulation of \cite{Barrett:2008wh}, for which condition \ref{eq-exists2} is satisfied.

Applying our theorem above to the EPRL partition function on $\CC_{33}$ we see that the support of the partition function asymptotically is contained in those geometries that are the boundary of an internal geometry that has discrete curvature values on the internal face. This is thus a perfect example of the accidental curvature constraints in action.

We can understand this from the point of view of the sum over spins, too. While the terms asymptotically separated from the stationary point are suppressed exponentially, this is no longer true for terms at finite distance. If the spins differ from the stationary conditions only by fixed finite amount then in the asymptotic limit the contribution is of the same order as from the stationary point. These terms of similar size as the dominant term will contribute to the asymptotic behaviour. Symbolically, while the asymptotics show that $f(\lambda (j_{\mbox{crit}} + 1))\ll f(\lambda j_{\mbox{crit}}) $, we can not conclude that $f(\lambda j_{\mbox{crit}} + 1)$ is also small. In fact we expect that $f(\lambda j_{\mbox{crit}} + 1) \sim f(\lambda j_{\mbox{crit}})$\footnote{In particular Magliaro and Perini in \cite{Magliaro:2011qm} discarded all such terms at a finite distance, resulting in a sparse sum that behaved as if only the dominant term were there.}.

This problem neatly demonstrates that the question of gluing asymptotic results together correctly in the interior is enormously subtle in the spin picture. In the holonomy formulation of spin foam models, the properties of the wave front sets given above handle all these subtleties for us.

%% file: AD-Discussion.tex
\section{Discussion}\label{sec:Discussion}

In this paper we have introduced the notion of the wave front set of a distribution into the study of spin foam partition function. As a result we could obtain a sufficient criterion for the partition function to be well defined and a necessary criterion for the partition function to be not suppressed in the limit of large boundary spins. This immediately begs two questions, how sufficient and how necessary are these criteria.

\subsection{The role of the Regge equations of motion}

We begin with the latter question. Our results give a necessary condition for the boundary variables to be in the wave front set. They need to occur as the boundary of an element of $\WW(\CC,\Gamma)$. If this were also sufficient for the boundary to be not suppressed these theories would not encode the Einstein equations in any form, all geometric manifolds would occur as semi-classical solutions in the amplitude.

For a discrete geometry, the additional condition we want a configuration to satisfy in order to be in the wave front set of the partition function, is for it to occur as the boundary of a Regge manifold. That is, there should be a continuation of the boundary data into the interior that satisfies the Regge equations of motion.

Our control on the structure of the distributions is not precise enough to establish whether or not this is the case. To do so we will need to develop the tools introduced here further. Further, we do not necessarily expect that the spin foam models studied here would satisfy them, due to the presence of many additional semi classical configurations that are not discrete geometries. In the case of the $\gamma = 0$ models we can also expect problems from the results of the vertex amplitude, which showed that the Regge action does not appear in the geometric sector. In fact this can be seen as the reason that these models do not suffer from flatness, pointing to a possibly severe conundrum.

Our analysis above focused on the geometric sector, where $s(v) = +$. Understanding the dynamics of the sector with $s(v) = BF$, and their coupling would be crucial to truly understand the behaviour of the partition function.

\subsection{Extension problems and regularisation}
On the other hand we gave a sufficient criterion for the partition function to be well defined. Recent results by Bonzom and Dittrich suggest that this might be far to conservative. In \cite{Bonzom2013} they study the behaviour of the convolution of a single face amplitude with the BC E-functions, and find, surprisingly that this actually remains finite on much of the singular support. More generally choosing better behaved face amplitudes $\omega$ always enables us to make the model finite \cite{Perez2001b,Perez2001}. In \cite{Riello2013} it was also found that even if not finite, an important class of triangulations is only logarithmically divergent for the EPRL model.

If the model is infinite and requires regularisation, the wave front set analysis still tells us about the way in which it requires regularisation. In particular, we also learn what parts of the spin foam integrand are regular and should thus not be changed by the regularisation. The remaining ambiguities of the regularisation are then known as the extension problem. It is important to note that extending products of distributions to locations where they are not well defined in this way can potentially change the wave front set. However, the singular support does remain the same in this procedure.

\subsection{Conclusions}

As the singular support of the distribution does not change, the flatness problem, which arises from the equations imposed by the face amplitude, remains. More precisely, as long as the $g_{ev}$ are a geometric connection and the $g_{ef}$ are generated by twisted simple bivectors, the condition $g_{f} = \id$ will induce the accidental curvature constraints.

Following the geometricity discussion in the appropriate variables we now can understand that the flatness is not a result of taking incorrect limits, or of a mistake in the quantisation, but instead is due to the interaction of the discretisation and holonomisation of the geometric connection on the one hand, and the introduction of the Immirzi parameter on the other.

In order to avoid the accidental curvature constraints, we are thus forced to give up on the twisting (as in the ${\text{FK}^0}$ model), the simple form of the face amplitude (as with our modified $D(\tilde{g},g)$), or the geometricity of the connection. Note that the modified partition functions with $D$ does not allow for the $\SU(2)$ boundary Hilbert space, see the discussion in \cite{Dittrich2013}, whereas $\text{FK}^0$ suffers sever problem that the asymptotic of the vertex amplitude does not reproduce the Regge action.

Thus our results pose severe problems for the program of obtaining continuum geometries satisfying the Einstein equations via intermediate coarse triangulations. In the models under investigation, the curvature at the scale of a single vertex is severely restricted, this in turn means that the variation over the set of allowed geometries, weighted with the Regge action, will almost certainly not produce Regge manifolds. Furthermore we can not use single or few 4-simplex calculations to study, even perturbatively, the behaviour around flat space, we can only probe flat space directly. Any perturbation of the curvature has to come in multiples of $2 \pi/\gamma$.

In our view, the only way around this conclusion is to view the individual simplices as microscopic objects. In that case it is not the observable curvature that is restricted but the local, small scale curvature. In this context the exact flatness of the $\gamma = \infty$ models can actually be seen as a boon. It restricts the curvature at the microscopic single vertex scale, and thus makes the model much more regular.

On the methodological side, we have seen that the introduction of holonomy spin foam models, and wave front sets greatly facilitated understanding the underlying geometry of the composition of the individual building blocks. Understanding the wave front sets of the $E$ and $\omega$ was no more complicated than understanding the asymptotics of the vertex amplitude. The composition of these then is mostly done at the structural level, without needing detailed knowledge of the distributions involved. This gives us a very clear picture as to which aspects of the amplitude are responsible for which equations. The fact that the vertex amplitude asymptotics and the behaviour around a face appear on the same footing in this formulation allowed us for the first time to clearly understand the origin and the precise formulation of the flatness problem.

While conditional our results also are the first results of geometricity that pertain to the entire partition function, rather than to its individual weights. One way to view the wave front set calculus is that it allows us to understand how the asymptotic limits of various building blocks combine to give the asymptotic limit of the partition function built from them.

The big outstanding challenge at this point is to study the degree of sufficiency and necessity in our main result. This would allow us to obtain unconditional statements. One path towards this goal is to combine our analysis here with the spinor formulation of spin foam models. Wave front sets behave well under symplectic reduction, thus combining these to methods has the potential to further simplify and extend the analysis of these amplitudes.

%% file: AD-JHEP-Appendix1.tex

\section{Wave front set lemmas}\label{sec:WFS-Lemmas}

This section is devoted to derivation of main properties of the simplicity functions in EPRL and FK
models. We will also prove useful criteria for similar distributions.

\subsection{The wave front set of the EPRL E-function}
\label{sec-E-der}
Here we derive the E-function asymptotics.
Let us assume $\gamma>1$. The derivation for the case $\gamma<1$ is similar and we
leave it to reader.

We need to analyze $E(g)$. We know
\begin{itemize}
 \item $E$ is invariant under $SU(2)$ action
\begin{equation}
 E(g)=E(hgh^{-1}),\ h\in \SU(2)\subset \Spin(4)
\end{equation}
\item $E$ is the solution of $CE=0$
where
\begin{equation}
C=(\gamma-1)\sqrt{C_++\frac{1}{4}}-(1+\gamma)\sqrt{C_-+\frac{1}{4}}+2
\end{equation}
where $C_\pm$ are the Casimirs of $\SU(2)_\pm$.
We have the coordinates on
$T^*\SU(2)$ given by
\begin{equation}
 (x,\wec{p})\colon x\in \SU(2), \wec{p}\in \su(2)\text{ left invariant vector fields}
\end{equation}
The Poisson brackets are as follows
\begin{equation}
 \{p_i,p_j\}=\epsilon_{ijk}p_k,\ \{p_i,f(x)\}=(L_if)(x),\ \{f_1(x),f_2(x)\}=0\ .
\end{equation}
The principal symbols of the Casimirs are given by $|\wec{p}^\pm|^2=\sum_i (p^\pm_i)^2$.
The principal symbol of $C$ is thus
\begin{equation}
 c=(\gamma-1)|\wec{p}^+|-(1+\gamma)|\wec{p}^-|\ .
\end{equation}
\item $E$ is also the solution of $\tilde{C}E=0$ where
\begin{equation}
 \tilde{C}=2\sqrt{C_++\frac{1}{4}}-(1+\gamma)\sqrt{L^2+\frac{1}{4}}+\frac{1}{2}(\gamma-1)
\end{equation}
where $L^2=(L_i^++L^-_i)^2$\ .
 The principal symbol of $\tilde{C}$ is
\begin{equation}
 \tilde{c}=2 |\wec{p}^+|-(1+\gamma)|\wec{p}^++\wec{p}^-|\ .
\end{equation}

\item From the equation
\begin{equation}
 E(g)=\sum d_e(k)\int dn \langle n|g^+|n\rangle^{(1+\gamma)k}\langle
n|g^-|n\rangle^{(1-\gamma)k}
\end{equation}
where $d_e(k)$ has at most polynomially growth,
we know that beside the set
\begin{equation}
 \{(g^+,g^-)\colon \exists \wec{n},\ g^\pm=e^{\pm
\alpha(1\mp\gamma)\wec{n}}\}\ ,
\end{equation}
$E(g)$ is a smooth function.
\end{itemize}

We have
\begin{theo}\label{theo-eprl-unique}
 The conical subset
$W$of $T^*Spin(4)\setminus\{0\}$ with the properties
\begin{itemize}
 \item $W$ is invariant under adjoint action of the diagonal $\SU(2)$, elements
of $W$ annihilate generators of this action
\item $W\subset\{c=0\}$, $W$ is invariant under the hamiltonian action of $c$
\item $W\subset\{\tilde{c}=0\}$, $W$ is invariant under the hamiltonian action
of $\tilde{c}$,
\end{itemize}
is the sum of sets of the type $W_\alpha^\pm$ defined by
\begin{align}
W^\pm_\alpha=\left\{\left(e^{\alpha(\wec{n},\wec{0})+\xi\hodge T_{-\gamma}(\wec{n},\wec{n})},\pm\chi
T_\gamma(\wec{n},\wec{n})\right)\colon
\chi>0,\xi\in {\mathbb R},|\wec{n}|=1\right\}
\end{align}
\end{theo}

Let us notice the following identities
\begin{align}
 W^\pm_{\alpha+2\pi}&=W^\pm_{\alpha}\\
W^\pm_{\alpha}&=W^\mp_{-\alpha}
\end{align}

\begin{proof}
The only co-vectors $(\wec{p}^+,\wec{p}^-)$ being the solution of $\tilde{c}=c=0$ by
Lemma \ref{lm-1} are of the form
\begin{equation}
 \wec{p}^\pm=(1\pm\gamma)\wec{n}\ .
\end{equation}
The action of $SU(2)$ is defined as
\begin{equation}
 g^\pm\rightarrow h^{-1}g^\pm h
\end{equation}
Its vector fields are spanned
for such a point $(g^+,g^-)$ by the vectors
\begin{equation}
 (\wec{m}-(g^+)^{-1}\akt \wec{m},\wec{m}-(g^-)^{-1}\akt\wec{m})\colon \wec{m}\in su(2)
\end{equation}
We have
\begin{equation}
 \left(\wec{m}-(g^+)^{-1}\akt\wec{m},\wec{m}-(g^-)^{-1}\akt\wec{m}\right)\perp (\wec{p}^+,\wec{p}^-)
\end{equation}
and thus by Lemma \ref{lm-2}
\begin{equation}
 g^\pm=e^{\alpha^\pm \wec{n}}
\end{equation}
Let us consider the flow generated by $c$ and $\tilde{c}$.
We have
\begin{equation}
 \{c^\pm,p^\pm_i\}=0,\ \{c^\pm,f(x^\pm)\}=2p^\pm_i (L_if)(x^\pm)
\end{equation}
and so
\begin{align}
 &\{c,p^\pm_i\}=0\\
&\{c,f(x^+,x^-)\}=\frac{(\gamma-1)p^+_i}{|p^+|}L^+_if-
\frac{(1+\gamma)p^-_i}{|p^-|}L^-_if
\end{align}
The vector field generated by $c$ on $T^*M-\{0\}$ is given by
\begin{equation}
 \frac{(\gamma-1)p^+_i}{|p^+|}L^+_i-
\frac{(1+\gamma)p^-_i}{|p^-|}L^-_i
\end{equation}
and by $\tilde{c}$
\begin{equation}
 \frac{2 p^+_i}{|p^+|}L^+_i-
\frac{(1+\gamma)(p^-_i+p^+_i)}{|p^-+p^+|}(L^-_i+L^+_i)
\end{equation}
Let us notice that on the set $(p^+,p^-)$ considered, the flow $c$
is given by
\begin{equation}
 -(1-\gamma)L^+_{\frac{n}{|n|}}-(1+\gamma)L^-_{\frac{-n}{|n|}}=(\gamma-1)L^+_{\frac{n}{|n|}}+(1+\gamma)L^-_{\frac{n}{|n|}}
\end{equation}
because $p^-=-(\gamma-1)n$.
The flow of $\tilde{c}$ is proportional to the flow of $c$.
So if the point $(e^{i\alpha_+ \wec{n}'},e^{i\alpha_- \wec{n}'};
T_\gamma(\wec{n}',\wec{n}')$ is in the set then
also
\begin{equation}
 \forall \xi,\quad(e^{(\alpha_- \wec{n}'-\xi(1-\gamma)\wec{n}')},e^{(\alpha_-
\wec{n}'+\xi(1+\gamma)\wec{n}')};T_\gamma(\wec{n}',\wec{n}'))
\end{equation}
is in this set. Finally $W$ ($SU(2)$ invariant) must consist of a sum of the
sets of the form
\begin{equation}
W^\pm_\alpha=\left\{\left(e^{\alpha(\wec{n},\wec{0})+\xi\hodge T_{-\gamma}(\wec{n},\wec{n})},\pm
T_\gamma(\wec{n},\wec{n})\right)\colon
\xi\in {\mathbb R},\wec{n}\in \su(2)\right\}
\qedhere
\end{equation}
\end{proof}

\subsubsection{Sublemmas}

\begin{lem}\label{lm-2}
Suppose for some vector $n\not=0$ holds
\begin{equation}
 \forall \wec{m}\in su(2),\ \left(\wec{m}-(g^+)^{-1}\akt\wec{m},\wec{m}-(g^-)^{-1}\akt\wec{m}\right)\perp ((1+\gamma)\wec{n},(1-\gamma)\wec{n})
\end{equation}
then $\exists \alpha^\pm$ such that $g^\pm=e^{i\alpha^\pm \wec{n}}$.
\end{lem}

\begin{proof}
\begin{equation}
 \forall \wec{m}\quad\tr \big[(1+\gamma)(\wec{m}-(g^+)^{-1}\akt\wec{m})\wec{n}+(1-\gamma)(\wec{m}-(g^-)^{-1}\akt\wec{m})\wec{n}\big]=0
\end{equation}
This is equivalent to
\begin{align}
\begin{split}
 &\wec{m}\cdot \big[(1+\gamma)(\wec{n}-g^+\akt\wec{n}+(1-\gamma)(\wec{n}-g^-\akt\wec{n})\big]=0\\
 &(1+\gamma)(\wec{n}-g^+\akt\wec{n})+(1-\gamma)(\wec{n}-g^-\akt\wec{n})=0
\end{split}\\
 &2\wec{n}+(\gamma-1)g^-\akt\wec{n}=(1+\gamma)g^+\akt\wec{n}
\end{align}
But
\begin{align}
 |\wec{n}|=|g^-\akt\wec{n}|=|g^+\akt\wec{n}|\\
2+(\gamma-1)=(1+\gamma)
\end{align}
and so in order to satisfy triangle inequality we have to have
\begin{equation}
 \wec{n}=g^-\akt\wec{n}=g^+\akt\wec{n}
\end{equation}
This is equaivalent to the statement
\begin{equation}
g^\pm=e^{\alpha^\pm \wec{n}}\qedhere
\end{equation}
\end{proof}

\begin{lem}\label{lm-1}
 The only vectors $\wec{p}^\pm$ satisfying
\begin{equation}\label{eq-lm-1}
 \begin{split}
 &2|\wec{p}^+|=(1+\gamma)|\wec{p}^++\wec{p}^-|\\
&(1+\gamma)|\wec{p}^-|=|1-\gamma||\wec{p}^+|
 \end{split}
\end{equation}
are $\wec{p}^\pm$ of the form
\begin{equation}
 \wec{p}^\pm=(1\pm\gamma)\wec{n}\ .
\end{equation}
\end{lem}

\begin{proof}
 We have from \eqref{eq-lm-1}
\begin{equation}
 \begin{split}
 &4(\wec{p}^+)^2=(1+\gamma)^2\left((\wec{p}^+)^2+2\wec{p}^+\cdot \wec{p}^-+(\wec{p}^-)^2\right)\\
&(1+\gamma)^2(\wec{p}^-)^2=(1-\gamma)^2(\wec{p}^+)^2
 \end{split}\label{pppm}
\end{equation}
and so
\begin{equation}
2(1-\gamma^2)(\wec{p}^+)^2=2(1+\gamma)^2 \wec{p}^+\wec{p}^-
\end{equation}
this is equivalent to
\begin{equation}
 [(1-\gamma)\wec{p}^+]\circ [(1+\gamma)\wec{p}^-]=\left|[(1-\gamma)\wec{p}^+]\right|^2
\end{equation}
since $\left|[(1-\gamma)\wec{p}^+]\right|=\left|[(1+\gamma)\wec{p}^-]\right|$ we have
\begin{equation}
 (1-\gamma)\wec{p}^+=(1+\gamma)\wec{p}^-
\end{equation}
what is equivalent to the thesis.
\end{proof}

\subsection{The coherent state kernel distribution}
\label{sec:coherent_state_dist_kernel}

We will prove

\begin{lem}\label{lem-FK}
 The wave front set of
\begin{equation}
 K_{n_i}(x,g_1,\ldots, g_n)=\sum_{\lambda\geq 0} e^{-i\lambda x}\langle n^1_1|g^1|n_2^1\rangle^{2\lambda
j_1}\cdots \langle n_1^n|g^n|n_2^n\rangle^{2\lambda j_n}
\end{equation}
satisfies
\begin{equation}
\begin{split}
 \WF(K_{n_i}){=}\Big\{&
\left(\chi,g_1,\ldots g_n; -\xi,
\wec{p}_1,\dots \wec{p}_n\right)\colon \\
&\chi=\sum_i 2j_i\ln\langle n^i_1|g_i|n^i_2\rangle,\
\xi>0,\ \forall_i\ \wec{n}^i_1=g_i\akt\wec{n}^i_2,\
\wec{p}_i=2j_i\xi \wec{n}_2^i \Big\}\cup\{0\}
\end{split}
\end{equation}
\end{lem}

Let us consider an operator
\begin{equation}
 {K}\colon C^\infty(S^1)\rightarrow C^\infty(\SU(2))
\end{equation}
given by an integral kernel distribution on $S^1\times \SU(2)$
\begin{equation}
 {K}(x,g)=\sum_{k\geq 0} e^{-ikx}\langle n|g|n\rangle^{k}
\end{equation}
We can easily check that
\begin{enumerate}
 \item $CV=0$ where
\begin{equation}
 C=i\frac{\partial}{\partial x}-\sqrt{L^2+\frac{1}{4}}-\frac{1}{2}
\end{equation}
\item ${K}$ is invariant with respect to transformations
\begin{equation}
\left\{ \begin{aligned}
  x\rightarrow x+\lambda\\ g\rightarrow e^{\lambda \wec{n}}g
 \end{aligned}\right\},\quad
 \left\{\begin{aligned}
  x\rightarrow x+\lambda\\ g\rightarrow ge^{\lambda \wec{n}}
 \end{aligned}\right\}
\end{equation}
\end{enumerate}
As before we can obtain conditions for $\WF$
\begin{itemize}
 \item $c=-p_x-|\wec{p}|=0$ and $\WF$ is invariant with respect to the flow
\begin{equation}
 -\partial_x-\frac{p^i}{|p|}L^i
\end{equation}
and $p_x,p^i$ is preserved.
\item $(p_x,p)\perp (1,g^{-1}\akt\wec{n})$ so
\begin{equation}
 p_x+\vec{p}\circ g^{-1}\akt\wec{n}=0
\end{equation}
\item $(p_x,\wec{p})\perp (1,\wec{n})$
\begin{equation}
 p_x+\wec{p}\circ g^{-1}\akt\wec{n}=0
\end{equation}
\end{itemize}
We know that $|\wec{p}|=-p_x\geq 0$
\begin{equation}
 \wec{p}\circ \wec{n}=|\wec{p}|
\end{equation}
since $|\wec{n}|=1$ we have $\wec{p}=-p_x \wec{n}$ and similarly $\wec{p}=-p_x g^{-1}\akt\wec{n}$.
The covector $\wec{n}$ is stabilized by $g$ so $g=e^{\chi \wec{n}}$. The flow of $C$ (and both actions) are
\begin{equation}
 -\partial_x-n^iL^i
\end{equation}
so if $(x,e^{\chi \wec{n}}, -\xi, \xi \wec{n})\in \WF$ then also
\begin{equation}
(x+\lambda,e^{(\chi+\lambda) \wec{n}}, -\xi, \xi \wec{n})\in \WF
\end{equation}
The singular support is
\begin{equation}
 {\rm sing}\ {K}=\{ \chi, e^{\chi \wec{n}}\}
\end{equation}
and thus the wave front set is given by
\begin{equation}
 \WF=\{(\chi,e^{i\chi \wec{n}}, -\xi, \xi \wec{n}), \xi>0\}\cup\{0\}
\end{equation}

We can now generalise to
\begin{equation}
 {K}_{n_1,n_2}(x,g)=\sum_j e^{-i2jx}\langle n|g_1^{-1}gg_2|n\rangle^{2j}
\end{equation}
where $\wec{n}_1=g_1\wec{n}$ and $\wec{n}_2=g_2\wec{n}$.

We can describe this wave front set in a geometric way
\begin{equation}
 \WF({K}_{n_1,n_2})=
\left\{ (\ln\langle n_1|g|n_2\rangle, g; -\xi, \xi \wec{n}_2)\colon \wec{n}_1=g \akt\wec{n}_2,\ |\wec{n}_2|=1,\ 
\xi>0\right\}\cup\{0\}
\end{equation}

This result easily generalises to many group variables since
\begin{equation}
\begin{split}
 K_{n_i}(x,g_1,\ldots, g_n)&=\sum_\lambda e^{-i\lambda x}\langle n^1_1|g^1|n_2^1\rangle^{2\lambda j_1}
\cdots \langle n_1^n|g^n|n_2^n\rangle^{2\lambda j_n}\\
&=\int \dd x_1\cdots \dd x_n\
\frac{1}{(2\pi)^{n-1}}
\delta_{2\pi}\left(x-\sum_i2j_ix_i\right)
\prod_i {K}_{n_1^i,n_2^i}(x_i,g_i)
\end{split}
\end{equation}
because
\begin{equation}
\delta_{2\pi}\left(x-\sum_i2j_ix_i\right)=\frac{1}{2\pi}
\sum_k e^{-ikx} \prod_i e^{ik2j_ix_i}
\end{equation}
We have from property \ref{pro-delta}
\begin{equation}
 \WF(\delta_{2\pi})=\left\{(x,x_1,\ldots x_n;p,p_1,\ldots p_n)\colon
x=\sum_i 2j_ix_i \text{ mod }2\pi,\ p_i=2j_ip\right\}\cup\{0\}
\end{equation}

Applying properties  \ref{pro-prod}, \ref{pro-res}, \ref{pro-int} we obtain 
\begin{equation}\label{eq:IntermediateSet}
\begin{split}
 \WF(K_{n_i})\subset\Big\{&
\left(\chi,g_1,\ldots g_n; -\xi,
\wec{p}_1,\dots \wec{p}_n\right)\colon \\
&\chi=\sum_i2j_i\ln\langle n^i_1|g_i|n^i_2\rangle,\
\xi>0,\ \forall_i\ \wec{n}^i_1=g_i\akt \wec{n}^i_2,\
\wec{p}_i=2j_i\xi \wec{n}_2^i \Big\}\cup\{0\}
\end{split}
\end{equation}
Now since
\begin{equation}
 K_{n_i}=\frac{1}{1-e^{-i x}\langle n^1_1|g^1|n_2^1\rangle^{2
j_1}\cdots \langle n_1^n|g^n|n_2^n\rangle^{2 j_n}}
\end{equation}
we know that the singular support is
\begin{equation}
 {\rm sing}\ K_{n_i}=\Big\{
\left(\chi,g_1,\ldots g_n\right)\colon 
\chi=\sum_i2j_i\ln\langle n^i_1|g_i|n^i_2\rangle,\
\forall_i\ \wec{n}^i_1=g_i\akt \wec{n}^i_2\Big\}
\end{equation}
This allows us to determine $\WF(K_{n_i})$ uniquely. This is due to the fact that for $g$ to be in the singular support, there needs to be a non zero direction above it in the wave front set, however, there is only one direction per $g$ in the sets of lemma \eqref{eq:IntermediateSet}. This proves the result.

\subsection{Elliptic operators}\label{sec:ellipticity}

Let $f\in C^\infty({\mathbb R}^n)$.
For a conical subspace of $V\subset {\mathbb R}^n$ we say that $f\in S^\rho$ if 
\begin{equation}
 \forall_{n_i\geq 0}\ \exists C>0\ \forall p \colon \left|\frac{\partial^{\sum_i n_i}f}{\partial p_1^{n_1}\cdots\partial p_n^{n_n}}\right|\leq C\left(1+\sum_i p_i^2\right)^{\frac{\rho-\sum_i n_i}{2}}
\end{equation}
Let us consider the operator $P: C^\infty((S^1)^n)\rightarrow C^\infty((S^1)^n)$ with integral kernel
\begin{equation}
 P(x_1,\ldots,x_n;y_1,\ldots y_n)=\sum_{n_1,\ldots, n_n\in{\mathbb Z}} f(n_1,\ldots,n_n) e^{-i\sum_i n_i(x_i-y_i)}
\end{equation}
The operator $P$ can be written using Fourier transform as
\begin{equation}
 f\left(i\frac{\partial}{\partial x_1},\ldots, i\frac{\partial}{\partial x_n}\right)\ .
\end{equation}
We can then have an important property proved in \cite{Strichartz1972}

\begin{pro}[Function \cite{Strichartz1972}]
For $f\in S^\rho({\mathbb R}^n)$ the operator
\begin{equation}
 P=f\left(i\frac{\partial}{\partial x_1},\ldots, i\frac{\partial}{\partial x_n}\right)
\end{equation}
is a classical pseudodifferential operator \cite{lars} with the symbol
\begin{equation}
 f(p_1,\ldots,p_n)+O(x,|p|^{\rho-1})
\end{equation}
\end{pro}

Together with the definition of microlocal invertibility (ellipticity) \cite{lars, grigis1994microlocal} we have

\begin{pro}[Ellipticity]\label{pro:ellip}
Assume that $f\in S^\rho({\mathbb R}^n)$ and for some conical neighbourhood $V$ of $(p_1,\ldots, p_n)\in{\mathbb R}^n$ holds
\begin{equation}
 \exists C>0\ \forall (p_1'\ldots,p_n')\in V,\ f(p_1'\ldots,p_n')\geq C\left(1+|p'|^2\right)^{-\rho/2},
\end{equation}
then for every $(x_1,\ldots, x_n)\in (S^1)^n$ and any distribution $A$ on $(S^1)^n$
\begin{equation}
 (x_1,\dots,x_n,p_1,\ldots p_n)\in \WF(A)\Rightarrow (x_1,\dots,x_n,p_1,\ldots p_n)\in \WF(PA)
\end{equation}
where $P$ is defined as
\begin{equation}
 P=f\left(i\frac{\partial}{\partial x_1},\ldots, i\frac{\partial}{\partial x_n}\right)
\end{equation}
\end{pro}

\subsection{Smooth distributions}
\label{sec:smooth}

In this section we describe a useful criterium for determining if a distribution is a smooth function \cite{lars,grigis1994microlocal}. On the manifold $M$ we choose any smooth nonvanishing Lebegue measure. Let $Q$ be an elliptic positive operator of first order \cite{lars,grigis1994microlocal}. For example $\sqrt{\sum L^2_a}$ for the product of group manifolds. Let us also introduce an orthonormal basis $\psi_i$ of smooth functions on the manifold $M$.
Then for the distribution $A$ we have that
\begin{equation}
 A\in C^\infty(M)\Leftrightarrow \forall_{n\geq 0}\ \exists C>0\ \forall_i\ |\la A | Q^n\psi_i\ra|\leq C
\end{equation}
Thus $A$ is not smooth if we can find a sequence $\tilde{\psi}_i$ of different eigenvectors of $Q$, such that
\begin{equation}
 \la A|\tilde{\psi}_n\ra\not=O(n^{-\infty})
\end{equation}

%% file: AD-JHEP-Appendix2.tex

\section{The geometry of projected spin network boundary states}\label{sec:ProjectedSpinNetworkBoundary}

In order not to overburden the paper with notation we developed the wave front set above for the case of the universal boundary Hilbert space \cite{Dittrich2013}. We also left the precise geometricity of the boundary implicit, focusing instead on the bulk equations of motion. In this appendix we will briefly sketch the boundary geometry that arises in the case of projected spin network boundaries. This example illustrates the concepts required to generalize our analysis to arbitrary other boundaries, including, in particular, the $\SU(2)$ boundary Hilbert space. 

\subsection{Wave front set of trimmed boundaries}

To obtain a projected spin network boundary we need to consider the partition function on a trimmed complex. That is, we have a dual 2-complex $\CC$ with a boundary graph dual to the boundary which ``cuts'' faces and edges touching the boundary. This boundary graph consists of edges $e \in \Gamma_e$ and vertices $v \in \Gamma_v$. Every $v \in \Gamma_v$ is the endpoint of a half-edge $e \in E$ touching the boundary. Edges in $\Gamma_e$ contain only one face. For more details see the discussion in \cite{Dittrich2013}.

The wave front set analysis extends to this case. The only new ingredient is the wave front set of the "square root" of the $E$ function, called $F$ in the preceding papers. For the case of EPRL we can choose $E$ such that $E = F$, and the extension really is straightforward. It is this case we will discuss here, leaving the extension to other models for future work.

In our case, the projected spin network space is given by $L^2(\Spin(4)^{\Gamma_e}/\SU(2)^{\Gamma_v})$, with $\SU(2) \subset \Spin(4)$ denoting the diagonal subgroup. For $\omega = \delta$ the part of the spin foam integrand associated to a single boundary edge $e'$ that is part of a trimmed face $f = (\dots,v_2,e_2,v',e',v'',e_3,v_4,e_4,\dots)$ takes the form
\begin{equation}
 \delta(\dots g_{v_2e_2} g_{e_2f} g_{v'e'v''} g_{e_3f} g_{e_3v_4} \dots) F(g_{e_2f}) F(g^{-1}_{e_3f}).
\end{equation}
Note that generically we have $F(g) \neq F(g^{-1})$, for the EPRL case that we will discuss, we still have $F = E$ and in order to not have to track inverses here we can write
\begin{equation}\label{eq:trimmed face}
 \delta(\dots g_{v_2e_2} g_{e_2f} g_{v'e'v''} g_{e_3f} g_{e_3v_4} \dots) E(g_{e_2f}) E(g_{e_3f}).
\end{equation}

This can then be convoluted with a coherent state of the form\footnote{It is geometrically more transparent to work with states that do not satisfy the gauge symmetry at the vertices, this is automatically implemented by the partition function anyways.} \[\la n_{v'e'}^+ |g^+_{v'e'v''} |n^+_{e'v''}\ra^{2j_e^+} \la n_{v'e'}^- |g^-_{v'e'v''} |n^-_{e'v''}\ra^{2j_e^-} \in L^2(\Spin(4)^{\Gamma_e})\ ,\] and by the same argument as given in section \ref{sec:large j limit}, the large $j$ limit of the amplitude with this state is then governed by the wave front set of the partition function.

To obtain the wave front set of \eqref{eq:trimmed face} we introduce the covectors $ p^{v''}_{e'} = p_{v'e'v''}$, $- p^{v'}_{e'} = p^{(f)}_{e_2f}$, and we have as before that $p^{(f)}_{v_2e_2} = p^{e_2}_{fv_2}$, $p^{(f)}_{e_3f} = p^{e_3}_{fv_4}$ and $p^{(f)}_{e_3 v_4} = p^{v_4}_{e_3e_4}$. These satisfy the transport equations
\begin{align}
&p^{v_2}_{e_2e_1} = \act g_{v_2e_2} p^{e_2}_{fv_2},  & \; p^{e_2}_{fv_2} = g_{e_2 f} \act p^{v'}_{e'},\nn\\
&p^{v'}_{e'} = - g_{v' e' v''} \act p^{v''}_{e'},  & \; p^{v''}_{e'} = g_{e_3f} \act p^{e_3}_{fv_4}, \nn\\
&p^{e_3}_{fv_4} = g_{e_3 v_4} \act p^{v_3}_{e_3e_4}\ . & 
\end{align}
As the $g_{ef}$ have to stabilise their own bivectors they drop out and we obtain $p^{e_2}_{fv_2} = p^{v'}_{e'}$ and $p^{v''}_{e'} = p^{e_3}_{fv_4}$ and thus more simply
\begin{align}
&p^{v'}_{e'} = g_{e_2v_2} \act p^{v_2}_{e_2e_1},  & \; p^{v''}_{e'} = g_{e_3 v_4} \act p^{v_3}_{e_3e_4},\nn\\
&p^{v'}_{e'} = - g_{v' e' v''} \act p^{v''}_{e'}\ . & 
\end{align}
Thus we have that these satisfy simplicity and closure, that is we have 
\begin{align}
N^0 \cdot T_{- \gamma} p^{v''}_{e'} = N^0 \cdot T_{- \gamma} p^{v'}_{e'} = 0\ ,\nn\\
\sum_{e:\ v' \in e \in \Gamma_e} p^{v'}_e = 0.
\end{align}
Thus, after reconstructing the geometric 4-simplices at the vertex, these boundary bivectors capture the boundary of the geometric 4-simplices in a common frame. If the interior set is $\gamma$ twisted, the boundary will be $\gamma$ twisted, too.

The geometric meaning of the group element $g_{v'e'v''}$ is much less clear though. The condition $p^{v'}_{e'} = - g_{v' e' v''} \act p^{v''}_{e'}$ only fixes this group element up to two angles, it is preserved by the transformation $g'_{v'e'v''} = \exp(\phi^1 p^{v'}_{e'} + \phi^2 \hodge p^{v'}_{e'}) g_{v'e'v''}$. These angles have to be fixed by the wave front set condition of the delta function. This tells us that
\begin{equation}\label{eq:BoundaryFaceId}
 g_{v'e'v''}^{-1} = g_{e_3f} g_{e_3v_4} g_{v_4e_4} \dots g_{v_2 e_2} g_{e_2 f}.
\end{equation}

To clarify both the geometric meaning of $g'_{v'e'v''}$ and reiterate the interpretation of the boundary covector variables we can again construct a solution from a given geometry.

\subsection{Bivectors and holonomies from geometry, Part 2: The boundary edition}

On the boundary vertices $v'\in\Gamma_v$ we have a tetrahedron $\tau^{v'}$. As the vertex $v'$ cuts the internal edge $e$ that it touches in half, this is identified with the tetrahedron $\tau^e$ of the interior complex. In the example of the face considered above that would be the edge $e_2$, and we have $N^0 \cdot \tau^{v'} = 0$ and 
\begin{equation}
\tau^{v'} G_{e_2v_2} = \tau^{v_2}_{e_2}.
\end{equation}

As we want the orientation of the boundary and the interior 4-simplex to match we also have
\begin{equation}
N^0 =  G_{e_2v_2} N^{v_2}_{e_2}.
\end{equation}

\begin{figure}[htp]
 \centering
 \includegraphics[scale=.70]{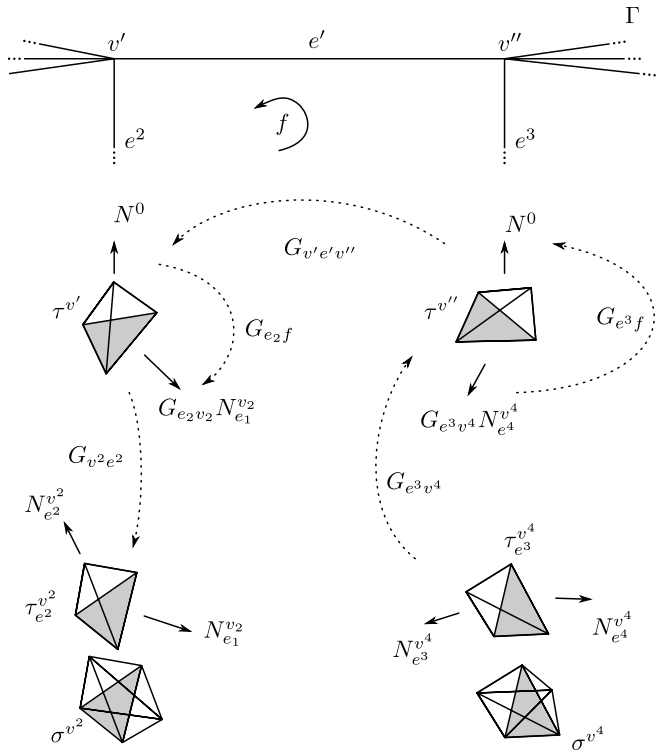}
 \caption{The geometric quantities at a boundary edge $e'$, in between boundary vertices $v'$ and $v''$ in a face $f=(\dots,v_2,e_2,v',e',v'',e_3,v_4,e_4,\dots)$.}
 \label{fig-data3}
\end{figure}

The boundary tetrahedron $\tau^{v'}$ has triangles $t^{v'}_{e'}$ with outward area normals $A^{v'}_{e'}$ in the direction of the boundary edge $e'$ connecting the boundary vertices $v'$ and $v''$. We can then also introduce the boundary connection $G_{v'e'v''}$ by requiring

\begin{align}\label{eq:Bound-Gvev-def}
N^0 &= G_{v'e'v''} N^{0},\nn\\
t^{v'}_{e'} &= G_{v'e'v''} t^{v''}_{e'},\nn\\
-A^{v'}_{e'} &= G_{v'e'v''} A^{v''}_{e'}.
\end{align}

The elements of the type $G_{e_2 f}$ that occur after a trimmed edge are introduced as

\begin{align}\label{eq:Bound-Gef-def}
G_{e_2v_2} N^{v_2}_{e_1} &= G_{e_2f} N^{0},\nn\\
t^{v'}_{e'} &= G_{e_2f} t^{v'}_{e'}.\nn\\
\end{align}

With these data we can define bivectors as
\begin{equation}
B^{v'}_{e'} = \hodge N^0 \wedge A^{v'}_{e'}.
\end{equation}
These geometric elements along with their relationships are illustrated in figure \ref{fig-data3}.

It is now easy to check that, if we can lift them to $\Spin(4)$, they satisfy the wave front set equations for $\gamma = 0$, in particular the bivectors are simple, satisfy the parallel transport condition and closure, and the product of group elements around a face is $\id$ as is easily seen by following the triangle dual to the face, and the normal orthogonal to it.

The question then arises to what degree one can invert this construction. After reconstructing the interior 4-simplices, and assuming they are geometric we directly obtain $g_{ev}$ that parallel transport the triangle (rather than just the bivector), and the $g_{ef}$ stabilize the triangle orthogonal to them by construction. Thus the condition \eqref{eq:BoundaryFaceId} immediately tells us that the $g_{v'e'v''}$ in the wave front set has to cover an element $G_{v'e'v''}$ that satisfies
$t^{v'}_{e'} = G_{v'e'v''} t^{v''}_{e'}$. If $G_{v'e'v''}$ is in the subgroup that satisfies $G_{v'e'v''} N^0 = N^0$, then $G_{v'e'v''}$ is simply the geometric connection of the boundary 3-geometry, and the group element used in defining Regge boundary states in \cite{Barrett:2009gg,Barrett:2009mw,Hellmann:2010nf}.

However this latter condition can indeed fail. To see how, note that the $g_{ef}$ have the non-geometric symmetry explained in section \ref{sec:Geom from Hol and Biv} in equation \eqref{eq:Symmetry-U(1)}. We had to use this symmetry to bring the $g_{ef}$ to the geometric ones. That this is always possible is ensured by equation \eqref{eq:FaceAngles}. If we use the same trick as in section \ref{sec:Geom from Hol and Biv}, and momentarily go to the frame in which all bivectors associated to the face $f$ are parallel to some bivector $B_f$, and thus $G_{v'e'v''} = \exp_{\SO(4)}(\phi^1 B_f + \phi^2 \hodge B_f)$ we can derive the analogous result in the presence of a boundary as
\begin{align}\label{eq:FaceAnglesBoundary}
\Theta_f + \sum_{e \in f} \xi_{ef} &= - \phi^2 \mod 2\pi\nn\\
\gamma \sum_{e \in f} \xi_{ef} &= \gamma \sum_{e \in f} \xi_{ef} - \phi^1 \mod 2\pi
\end{align}

The second line, implying $\phi^1 = 2\pi \N$ is simply the already stated fact that the triangle around the face is stabilized by the $G_{v'e'v''}$. The first line though shows that now the U$(1)$ symmetry can generally not set the $g_{ef}$ equal to the interior dihedral angle. The $g_{v'e'v''}$ acts as a boundary dihedral connection that adds another, arbitrary, deficit angle $\phi^2$ that needs to be compensated by the $g_{ef}$.

%% file: AD-JHEP-Appendix3.tex

\section{BF theory and wave front set conditions}
\label{sec:BF}

In this appendix we recall the notions of homological algebra used in the main body of the text to solve the BF equations of motion.

\subsection{Local homology}

For a configuration in ${\mathcal W}(\CC,\Gamma)$ of the BF partition function, parallel transports around any face is trivial.

We can thus introduce the set of elements $h_{ef}$ such that for every face
\begin{equation}
 g_{ev}g_{ve'}=h_{ef}h_{e'f}^{-1}
\end{equation}
In order to define the relevant operators we need to introduce an orientation on the edges as well. This can be encoded by writing the edges as an ordered pair of vertices, $e = (v,v')$. If an edge $e = (v,v')$ is in a face $e \in f$, their orientations can agree, meaning $(v,e,v') \subset f$, or disagree if instead $(v',e,v) \subset f$. We introduce signs $\epsilon_{ef}$ which are $+$ if the orientation of $e$ and $f$ agree, and $-$ otherwise. We further introduce signs $\epsilon_{ve}$, which are $+$ if $v$ is the target of $e$, that is, $e = (v,v')$, and $-$ if instead $e = (v',v)$, that is, if v is the source.

For $\CC$ we can introduce a boundary operator in local homology \cite{Barrett:2008wh}
\begin{equation}
 d^\CC_2\colon C(\CC_f)\rightarrow C(\CC_e),\quad d^\CC_2(\mathcal{F})(e)=\sum_{f\ni e} \epsilon_{ef} h_{ef}\rhd \mathcal{F}(f)
\end{equation}
where $C(\CC_f)$ and $C(\CC_e)$ are linear vector spaces generated by sets of $\su(2)$ labelled by $\CC_f$ and $\CC_e$. That is, an element $\mathcal{E} \in C(\CC_e)$ is a set of convectors $\{p_e\}$ with $e$ ranging over $\CC_e$, and $\mathcal{E}(e) = p_e$. Note that the covectors $p_e$ defined this way satisfy the parallel transport condition.

Similarly
\begin{equation}
 d^\CC_1: C(\CC_e)\rightarrow C(\CC_v),\quad d^\CC_1(\mathcal{E})(v)=\sum_{e\ni v} \epsilon_{ve} g_{ve}\rhd \mathcal{E}(e)
\end{equation}
Again the sign depends on the orientation of the edge.

The homology groups with local coefficients are defined as
\begin{equation}
 H_i(\CC,\su(2))=\ker d^\CC_i/{\rm im}\ d^\CC_{i+1}
\end{equation}
where $d_3=0$, $d_0=0$.

Similar construction applied to $\Gamma$ with convention that $C(\Gamma_f)=\{0\}$. We can thus identify
\begin{equation}
 H_1(\Gamma,\su(2))=\ker d^\Gamma_1\ ,
\end{equation}
which are sets of bivectors $p_{e}$ such that
\begin{equation}\label{eq:Boundary-closure}
 \forall_{v\in\Gamma_v} \sum_{e\ni v} \epsilon_{ve} g_{ve}\rhd p_{e}=0
\end{equation}

\subsection{Inclusion}

We can introduce the inclusion operators 
\begin{equation}
  {\rm incl_f}:\{0\}\rightarrow C(\CC_f),\ {\rm incl_e}:C(\Gamma_e)\rightarrow C(\CC_e),\  {\rm incl_v}:C(\Gamma_v)\rightarrow C(\CC_v)
\end{equation}
These are defined by 
\begin{equation}
({\rm incl_e} \mathcal{E})(e) =
\begin{cases}
\mathcal{E}(e)& e \in \Gamma_e\\
0 & \text{otherwise}
\end{cases}
\end{equation}
and similarly for $f$ and $v$.

This map descends to a map on homology
\begin{equation}
 [{\rm incl}]_i: H_i(\Gamma,\su(2))\rightarrow H_i(\CC,\su(2))
\end{equation}

In $H_1$ we can introduce the variables
\begin{equation}
 p_{ve}=\epsilon_{ve} p_{e}
\end{equation}
depending on the orientation. The conditions for $p_{ev}$ that arise from the wave front set are the following: There exists a set of $p_f$ such that under the map $d_2^\CC$ their image $p_e$ ($p_{ev}$) satisfies the closure condition on the internal edges (is zero there).

Using the inclusion map the condition on $p_f$ is that their image $p_{e}$ belong to the image of ${\rm incl}$. The set of valid boundary data is thus
\begin{equation}
 {\rm incl_e}^{-1}({\rm im}\ d_2^\CC)
\end{equation}

In terms of sets we can say that a valid boundary is a set of \[\{p_e\},\ e \in \Gamma\] satisfying \eqref{eq:Boundary-closure} such that the set \[\{p'_e\},\ e \in \CC\] defined as  \[p'_e = p_e\ \forall e \in \Gamma,\ p'_e = 0 \text{ otherwise}\] is of the form \[\{ p'_e = \sum_{f\ni e} \epsilon_{ef} h_{ef} p_f\}\] for some $p_f$.

As the condition \eqref{eq:Boundary-closure} is homological, we can view valid boundary data $s$ as elements of the first homolgy, that is, $s \in H_1(\Gamma,\su(2))$. Then the other two conditions, parallel transport and closure, can be written homologically as well as
\[[{\rm incl}]_1(s)=0\ .\]

Let us unpack this last line. Equality in $H_1(\CC,\su(2))$ is up to elements in ${\rm im}\ d_2^\CC$. That is, there must exist $p_f$, such that \[\{p'_e: \ p'_e = p_e \forall e \in \Gamma, \ p'_e = 0 \text{ otherwise} \}\] differs from \[\{p'_e \in \CC, p'_e = 0\}\] exactly by $\tilde p_e = \sum_{f \ni e} h_{ef} p_f $. This, however is just the condition explained before.

To summarise, we have that $s$ is a valid boundary configuration if 
\begin{equation}
 s\in \ker d_1=H_1(\Gamma,\su(2))\subset C(\Gamma_e),\quad [{\rm incl}]_1(s)=0,
\end{equation}
or, even more concisely, the admissible configurations $p_{ve} = -p_{ev} = \epsilon_{ve} p_e$ are those belonging to
\begin{equation}
[{\rm incl}]_1^{-1}(0)
\end{equation}
understood as a subset of $C(\Gamma_e)$.